\documentclass[10pt]{article}
\usepackage{amsthm,amsmath,amssymb}
\usepackage[linesnumbered,lined,boxed,ruled,vlined,algo2e]{algorithm2e}
\usepackage[pdftex,breaklinks,colorlinks,
linkcolor=blue,
citecolor=blue,
urlcolor=blue]{hyperref}

\newtheorem{theorem}{Theorem}[section]
\newtheorem{lemma}[theorem]{Lemma}
\newtheorem{corollary}[theorem]{Corollary}
\theoremstyle{definition}
\newtheorem{defn}{Definition}
\newtheorem{rem}{Remark}
\newcommand{\keywords}[1]{\bigskip \par\noindent
{\small{\em Keywords\/}: #1}}

\title{Neighborhood covering and independence on\\two superclasses of cographs}

\author{Guillermo Dur{\'a}n\thanks{CONICET, Argentina, Instituto de C{\'a}lculo and Departamento de Matem{\'a}tica, Facultad de Ciencias Exactas y Naturales, Universidad de Buenos Aires, Buenos Aires, Argentina, and Departamento de Ingenier\'{\i}a Industrial, Facultad de Ciencias F{\'{\i}}sicas y Matem{\'a}ticas, Universidad de Chile, Santiago, Chile. E-mail: \href{mailto:gduran@dm.uba.ar}{\tt gduran@dm.uba.ar}.} \and
        Mart\'{\i}n D.\ Safe\thanks{Instituto de Ciencias, Universidad Nacional de General Sarmiento, Los Polvorines, Buenos Aires, Argentina. E-mail: \href{mailto:msafe@ungs.edu.ar}{\tt msafe@ungs.edu.ar}.} \and
        Xavier S.\ Warnes\thanks{Instituto de C{\'a}lculo and Departamento de Matem{\'a}tica, Facultad de Ciencias Exactas y Naturales, Universidad de Buenos Aires, Buenos Aires, Argentina. E-mail: \href{mailto:xwarnes@dc.uba.ar}{\tt xwarnes@dc.uba.ar}.}
}
 
\date{December 31, 2015}

\newcommand{\BigO}{\mathcal{O}}
\newcommand{\NP}{$\mathcal{NP}$}
\newcommand{\pn}[1]{\rho_{\mathrm n}(#1)}
\newcommand{\an}[1]{\alpha_{\mathrm n}(#1)}
\newcommand{\at}[1]{\alpha_{2}(#1)}
\begin{document}

\maketitle

\begin{abstract}
Given a simple graph $G$, a set $C \subseteq V(G)$ is a neighborhood cover set if every edge and vertex of $G$ belongs to some $G[v]$ with $v \in C$, where $G[v]$ denotes the subgraph of $G$ induced by the closed neighborhood of the vertex $v$. Two elements of $E(G) \cup V(G)$ are neighborhood-independent if there is no vertex $v\in V(G)$ such that both elements are in $G[v]$. A set $S\subseteq V(G)\cup E(G)$ is neighborhood-independent if every pair of elements of $S$ is neighborhood-independent. Let $\rho_{\mathrm n}(G)$ be the size of a minimum neighborhood cover set and $\alpha_{\mathrm n}(G)$ of a maximum neighborhood-independent set. Lehel and Tuza defined neighborhood-perfect graphs $G$ as those where the equality $\rho_{\mathrm n}(G') = \alpha_{\mathrm n}(G')$ holds for every induced subgraph $G'$ of $G$.

In this work we prove forbidden induced subgraph characterizations of the class of neighborhood-perfect graphs, restricted to two superclasses of cographs: $P_4$-tidy graphs and tree-cographs. We give as well linear-time algorithms for solving the recognition problem of neighborhood-perfect graphs and the problem of finding a minimum neighborhood cover set and a maximum neighborhood-independent set in these same classes.
\end{abstract}


\section{Introduction}
\label{sec:intro}
A graph is \emph{perfect} if, for every induced subgraph, the maximum size of a clique equals the minimum number of colors needed to color its vertices such that no two adjacent vertices have the same color. One of the most celebrated results in the last fifteen years in Graph Theory is without a doubt the characterization by forbidden induced subgraphs of the class of perfect graphs. This characterization was proved by Chudnovsky, Robertson, Seymour and Thomas in 2002 \cite{ChudnovskyRoberstonSeymourRobinSPGT2006}, settling affirmatively a conjecture posed more than 40 years before by Berge \cite{berge1961farbung}. The minimal forbidden induced subgraphs of perfect graphs are the chordless cycles of odd length having at least $5$ vertices, called \emph{odd holes} $C_{2k+1}$, and their complements, the \emph{odd antiholes} $\overline C_{2k+1}$.

During the nearly half a century in which this characterization remained a conjecture, many graph classes were defined analogously to perfect graphs by the equality of two parameters (e.g.\ clique perfect graphs \cite{guruswami_algorithmic_2000}, coordinated graphs \cite{bonomo2007coordinated}, neighborhood-perfect graphs \cite{lehel_neighborhood_1986}). 

Neighborhood-perfect graphs were defined in 1986 by Lehel and Tuza \cite{lehel_neighborhood_1986}, by the equality of two parameters for all induced subgraphs. Given a simple graph $G$, a set $C \subseteq V(G)$ is a \emph{neighborhood-covering set} (or \emph{neighborhood set}) if each edge and each vertex of $G$ belongs to some $G[v]$ with $v \in C$, where $G[v]$ denotes the subgraph of $G$ induced by the closed neighborhood of the vertex $v$. Two elements of $E(G) \cup V(G)$ are \emph{neighborhood-independent} if there is no vertex $v\in V(G)$ such that both elements are in $G[v]$. A set $S \subseteq V(G)\cup E(G)$ is said to be a \emph{neighborhood-independent set} if every pair of elements of $S$ is neighborhood-independent. Let $\rho_{\mathrm n}(G)$ be the size of a minimum neighborhood-covering set and $\alpha_{\mathrm n}(G)$ of a maximum neighborhood-independent set. Clearly, $\rho_{\mathrm n}(G) \geq \alpha_{\mathrm n}(G)$ for every graph $G$. When $\rho_{\mathrm n}(G') = \alpha_{\mathrm n}(G')$ for every induced subgraph $G'$ of $G$, $G$ is called a \emph{neighborhood-perfect} graph. Since odd holes and odd antiholes are not neighborhood-perfect (\cite{lehel_neighborhood_1986}), the Strong Perfect Graph Theorem implies that all neighborhood-perfect graphs are also perfect. %For all undefined terminology, the reader is referred to \cite{west2001introduction}.

Neighborhood-perfect graphs have been characterized by forbidden induced subgraphs, when restricted to the classes of chordal graphs \cite{lehel_neighborhood_1986}, line graphs \cite{lehel_neighbourhood-perfect_1994} and cographs \cite{gyarfas_minimal_1996}. The characterizations presented here are an extension of this last result. Furthermore, Lehel and Tuza \cite{lehel_neighborhood_1986} proved that finding $\alpha_{\mathrm n}(G)$ and $\rho_{\mathrm n}(G)$ can be done in polynomial time if $G$ is a chordal neighborhood-perfect graph. If $G$ is strongly chordal, interval or a cograph (i.e., $P_4$-free), then linear-time algorithms that find the above mentioned parameters have been given \cite{Branstadt1997CliqueRDomCliqueRpackDuallyChordal,gyarfas_minimal_1996,lehel_neighborhood_1986}. On the other hand it was proven that this problem is $\mathcal{NP}$-complete over a class of split graphs with degree constraints \cite{chang_algorithmic_1993}. Although it follows from previous works (e.g.\ \cite{gyarfas_minimal_1996,lehel_neighbourhood-perfect_1994,lehel_neighborhood_1986}) that deciding whether a graph is neighborhood-perfect can be accomplished in polynomial-time if the input graph belongs to several different graph classes, the computational complexity of recognizing neighborhood-perfect graphs in general is unknown.

The work is organized as follows. In Section 2, we give some preliminary definitions and results, including an introduction to modular decomposition and the structure of the classes of $P_4$-tidy graphs and tree-cographs. In Section 3, we give formulas for $\alpha_{\mathrm n}$ and $\rho_{\mathrm n}$ for the join of two or more graphs and determine all the minimally non-neighborhood-perfect graphs whose complement is disconnected. In Section 4, we prove our structural results, which consist in minimal forbidden induced subgraph characterizations of the class of neighborhood-perfect graphs when restricted to the classes of $P_4$-tidy graphs and tree-cographs, respectively. In Section 5, we give our algorithmic results, which consist in linear-time recognition algorithms for neighborhood-perfectness of $P_4$-tidy graphs and tree-cographs, linear-time algorithms for computing $\alpha_{\mathrm n}$ and $\rho_{\mathrm n}$ for any given $P_4$-tidy graph or tree-cograph, and a proof that the problems of computing $\alpha_{\mathrm n}$ and $\rho_{\mathrm n}$ become \NP-hard for complements of bipartite graphs.

\section{Preliminaries} 
\label{sec:preliminaries}

Before we formulate the results, a few definitions that will be used later on are required. For all undefined terminology we refer to \cite{west2001introduction}. All graphs in this work are finite, undirected, and have no loops or multiple edges. Let $G$ be a graph. We shall denote by $V(G)$ its vertex set and $E(G)$ its edge set, by $m_G$ the size of $E(G)$ and by $n_G$ the size of $V(G)$ (omitting the subscript $G$ when it is clear by context). The complement of $G$ shall be denoted by $\overline{G}$, the neighborhood of a vertex $v$ by $N_G(v)$, and the closed neighborhood $N_G(v) \cup \{v\}$ by $N_G[v]$. We denote by $G[W]$ the subgraph of $G$ induced by $W \subseteq V(G)$. A vertex of $G$ is said to be \emph{pendant} if it is adjacent to exactly one vertex of $G$, \emph{universal} if it is adjacent to all other vertices, and \emph{simplicial} if its neighborhood is a clique. If $H$ is a graph, then $G$ is $H$-free if $G$ contains no induced subgraphs isomorphic to $H$. If $\mathcal{H}$ is a collection of graphs, then $G$ is $\mathcal{H}$-free if $G$ is $H$-free for all $H \in \mathcal{H}$. The chordless path of $k$ vertices is denoted $P_k$, and the chordless cycle of $k$ vertices $C_k$. A cycle is odd if it has an odd number of vertices. The graph $K_n$ is the complete graph of $n$ vertices. If $H$ is a graph and $t$ is a nonnegative integer, then $tH$ denotes the disjoint union of $t$ copies of $H$. The \emph{join} of two graphs $G_1=(V_1,E_1)$ and $G_2=(V_2,E_2)$ (where $V_1\cap V_2=\emptyset$) is the graph $G_1 \vee G_2=(V_1\cup V_2,E_1\cup E_2\cup\{uv\mid u\in V_1,v\in V_2\}$). The disjoint union of two graphs $G_1$ and $G_2$ is denoted by $G_1 + G_2$. The size of a set $S$ is denoted by $|S|$. We shall consider a \emph{clique} (resp.\ \emph{independent set}) to be a set of pairwise adjacent (resp.\ nonadjacent) vertices, and a \emph{2-independent set} to be a set of vertices such that for every pair of vertices of the set, there is no path of length $2$ or less that connects them in $G$. We shall use \emph{maximum} to denote of maximum size, while \emph{maximal} shall denote inclusion-wise maximal (analogously with \emph{minimum} and \emph{minimal}). Given a graph $G$, we shall denote by $\alpha(G)$ the size of any maximum stable set of $G$ and by $\at{G}$ the size of a maximum $2$-independent set of $G$. A set of vertices of $G$ is said to be a \emph{dominating set} if every vertex in $G$ belongs to the set or is adjacent to a vertex of the set. We shall denote by $\gamma(G)$ the size of any minimum dominating set of $G$. A \emph{matching} of a graph is a set of edges such that every node of $G$ is incident with at most one edge of the set. A \emph{vertex cover} of a graph $G$ is a set of vertices $W$ such that every edge of $G$ is incident to a least one vertex of $W$. We shall note by $\tau(G)$ the size of any minimum vertex cover of $G$ and by $\nu(G)$ the size of any maximum matching of $G$.

We say that a graph is \emph{co-connected} if its complement is connected. The \emph{anticomponents} of a graph are its maximal co-connected induced subgraphs (equivalently they are the complements of the components of $G$). This means that if $H_1\dots H_k$ are the anticomponents of $H$, then $H = H_1 \vee \dots \vee H_k$.

A graph $G$ is said to be \emph{chordal} if it is $C_k$-free, for every $k \geq 4$. A \emph{k-sun} (or \emph{trampoline} of order $k$) is a graph having $2k$ vertices $v_1,\dots,v_k,w_1,\dots,w_k$ such that $v_1\dots v_k$ is a cycle and every $w_i$ ($1\leq i \leq k$) has exactly two neighbors: $v_i$ and $v_{i+1}$ ($v_{k+1} = v_1$). A $k$-sun is said to be \emph{odd} if $k$ is odd.

 As stated in Section~\ref{sec:intro}, neighborhood-perfect graphs have been characterized by forbidden induced subgraphs restricted to the class of chordal graphs. The minimal forbidden induced subgraphs are exactly the odd suns defined above. 

 \begin{theorem}[\cite{lehel_neighborhood_1986}]\label{thm:np-chordal-forb}
   A chordal graph $G$ is neighborhood-perfect if and only if it contains no induced odd sun.
 \end{theorem}

A graph $G$ was defined to be minimally non-neighborhood-perfect in \cite{gyarfas_minimal_1996}, if $G$ is not neighborhood-perfect, but all proper induced subgraphs of $G$ are. We now state for future reference a result proven in \cite{gyarfas_minimal_1996} that determines exactly which graphs $G$are minimally non-neighborhood-perfect and have $\an{G} = 1$. 

\begin{theorem}[\cite{gyarfas_minimal_1996}]\label{teo:MinNonNPwithAN1Are3PyrAnd0Pyr}
If $G$ is a minimally non-neighborhood-perfect graph and $\an{G}=1$, then $G$ is a $3$-sun or $\overline{3K_2}$.
\end{theorem}

\subsection{Modular Decomposition}
Let $G$ be a graph. We shall say that a vertex $v$ of $G$ \textit{distinguishes} between two vertices $x$ and $y$ of $G$ if it is adjacent to one of them and nonadjacent to the other. A set $M$ of vertices shall be called a \textit{module} of $G$ if there is no vertex of $V(G)\setminus M$ that distinguishes any pair of vertices of $M$, or equivalently every vertex of $G$ not in $M$ is either adjacent to all vertices of $M$ or to none of them. The empty set, the singletons $\{v\}$ for each $v \in V(G)$ and $V(G)$ are the \textit{trivial modules} of $G$. A graph is said to be \textit{prime} if it has more than two vertices and it has only trivial modules (for example $P_4$ is a prime graph.) A nonempty module is \textit{strong} if, for every other module $M'$ of $G$, either $M'\subseteq M$, $M \subseteq M'$ or $M \cap M' = \emptyset$. The \textit{modular decomposition tree} $T(G)$ of a graph $G$ is a rooted tree having one node for each strong module of $G$ and such that a node $h$ representing a strong module $M$ has as its children the nodes representing the maximal strong modules of $G$ properly contained in $M$. Clearly the root of the tree represents the module $V(G)$, and every leaf one of the singletons $\{v\}$, for each $v\in V(G)$.\\
For each node $h$ of $T(G)$, we note the module represented by $h$ as $M(h)$. Note that, by construction, if we associate with each leaf the only vertex the module represents, then $M(h)$ corresponds to the set of leafs that have $h$ as an ancestor in $T(G)$.\\
For each node $h$ of $T(G)$, we denote the induced subgraph $G[M(h)]$ by $G[h]$ and call it the \textit{graph represented by h}. Each node of $T(G)$ that is not a leaf is a \textit{parallel}, \textit{series} or \textit{neighborhood} node, and called a \textit{P-node}, \textit{S-node} or \textit{N-node}, respectively. If $G[h]$ is disconnected, $h$ is a P-node; if $\overline{G[h]}$ is disconnected, $h$ is an S-node; and if both $G[h]$ and $\overline{G[h]}$ are connected, then $h$ is an N-node. Thus, if $h$ is an internal node of $T(G)$ and $h_1, \dots, h_k$ are the children of $h$ in $T(G)$, then one of the following conditions holds:
\begin{itemize}
    \item{If $G[h]$ is disconnected, then $h$ is a P-node and $G[h_1], \dots, G[h_k]$ are the components of $G[h]$.}
    \item{If $\overline{G[h]}$ is disconnected, then $h$ is an S-node and $G[h_1],\dots,G[h_k]$ are the anticomponents of $G[h]$.}
    \item{If $G[h]$ and $\overline{G[h]}$ are both connected, then $h$ is an N-node and $M(h_1),\dots,M(h_k)$ are the maximal strong modules of $G[h]$ properly contained in $M(h)$.}
\end{itemize}
In all of these cases, it holds that $\{M(h_1),\dots,M(h_k)\}$ is a disjoint partition of the nodes in $M(h)$ \cite{BuerDecomp1983,GallaiCographs1967}.

Let $h$ be a node of $T(G)$ and let $h_1,\dots,h_k$ be its children. We shall denote by $\pi(h)$ the graph having vertex set $\{h_1,\dots,h_k\}$ and such that $h_i$ is adjacent to $h_j$ if and only if there is some edge in $G$ joining a vertex of $M(h_i)$ and a vertex of $M(h_j)$. Since $M(h_i)$ and $M(h_j)$ are both modules of $G[h]$, then clearly there is an edge between them if and only if every vertex of $M(h_i)$ is adjacent to every vertex of $M(h_j)$. Hence $G[h]$ coincides with the graph that arises from $\pi(h)$ by successively substituting $h_i$ by $G[h_i]$, for each $h_i$. Note that each $\pi(h)$ must be a prime graph, since all $M(h_i)$ are maximal strong modules, and if $\pi(h)$ had a nontrivial module of more than one vertex, the modules $M[h_i]$ corresponding to these vertices would form a module of $G[h]$. We shall denote by $\pi(G)$ the set $\{ \pi(h) \colon \textit{h is an N-node of T(G)} \}$. The following result shows that every induced prime subgraph of a graph $G$ is also an induced subgraph of a graph in $\pi(G)$.

\begin{theorem}[\cite{Fouquet_OnSemiP4sparseGraphs_1997}]\label{theo:primeFreeGraphs}
Let $Z$ be a prime graph. A graph $G$ is $Z$-free if and only if each graph of $\pi(G)$ is $Z$-free.
\end{theorem}

In the rest of this work, we shall note $|V(G[h])|$ by $n(h)$, for every $h\in V(T(G))$. If $h$ is an N-node, then we shall note $|V(\pi(h))|$ by $n_{\pi}(h)$ and $|E(\pi(h))|$ by $m_\pi(h)$. A fact that will be used in what follows is that since $T(G)$ has $n$ nonadjacents and each internal node has at least two children, $T(G)$ must have less than $2n$ nodes. An important property that we shall use extensively is that the sum of $n_{\pi}(h)$, over all N-nodes $h$ of $T(G)$, is at most $2n$ \cite{Baumann96alinear}.

In this work we shall assume that each N-node $h$ of the modular decomposition tree $T(G)$ is accompanied by a description of the prime graph $\pi(h)$, by means of an adjacency list. There are linear-time algorithms to compute the rooted tree $T(G)$ \cite{CournierModDecomp1994,DahlhausDecomp2001,McConnellMDecomp1999,TedderSimpleModDecom2008}, moreover in \cite{Baumann96alinear} it is shown that the adjacency lists of each $\pi(h)$, for every N-node $h$, can be added also in linear-time. For a survey on the algorithmic aspects of modular decompositions, see \cite{HabibModDecompSurvey2010}.

\subsection{Structure of $P_4$-tidy graphs and tree-cographs}

We shall now introduce the two classes in which we will be studying neighborhood-perfectness. Both of these classes are generalizations of the class of cographs.

A graph $G=(V,E)$ is \emph{$P_4$-tidy} if for every vertex set $A$ inducing a $P_4$ in $G$ there is at most one vertex $v \in V \setminus A$ such that $G[A\cup\{v\}]$ contains at least two induced $P_4$'s. There is a structure theorem for $P_4$-tidy graphs that extends Seinsche's theorem in cographs, in terms of starfishes and urchins.

A \textit{starfish} is a graph whose vertex set can be partitioned in three disjoint sets, $S$, $C$ and $R$, where each of the following conditions holds:
 \begin{itemize}
     \item{$S = \{s_1,\dots,s_t\}$ is a stable set and $C = \{c_1,\dots,c_t\}$ is a clique, for $t\geq 2$.}
     \item{$R$ is allowed to be empty. If it is not, then all vertices of $R$ are adjacent to all vertices of $C$ and nonadjacent to all vertices of $S$.}
     \item{$s_i$ is adjacent to $c_j$ if and only if $i=j$.}
 \end{itemize}

An \textit{urchin} is a graph whose set can be partitioned intro three sets $S$, $C$ and $R$ satisfying the first two conditions stated above, but instead of satisfying the  third one, it must satisfy that:
\begin{itemize}
    \item{$s_i$ is adjacent to $c_j$ if and only if $i\neq j$.}
\end{itemize}

It is clear that urchins are the complement of starfish and vice versa. Given $G$, a starfish or urchin, and a partition $(S,C,R)$, we shall call $S$ the \textit{ends} of $G$, $C$ the \textit{body} of $G$, and $R$ the \textit{head} of $G$. A \textit{fat urchin} (resp.\ \textit{fat starfish}) arises from an urchin (resp.\ starfish), with partition $(S,C,R)$, by substituting exactly one vertex of $S\cup C$ by a $K_2$ or a $2K_1$.

\begin{theorem}[\cite{Giakoumakis_OnP4TidyGraphs_1997}] \label{theo:StructureOfP4Tidy}
If $G$ is a $P_4$-tidy graph, then exactly one of the following statements holds:
\begin{enumerate} 
     \item {$G$ or $\overline{G}$ is disconnected.}
     \item {$G$ is isomorphic to $C_5$, $P_5$, $\overline{P_5}$, a starfish, a fat starfish, an urchin, or a fat urchin.}
\end{enumerate}
\end{theorem}

Let $G$ be a $P_4$-tidy graph and $h$ be an N-node of $T(G)$, the modular decomposition tree of $G$. Theorem \ref{theo:StructureOfP4Tidy} implies that $\pi(h)$ must be isomorphic to $C_5$, $P_5$, $\overline{P_5}$, a prime starfish, or a prime urchin. Moreover, if $\pi(h)$ is isomorphic to a prime urchin or a prime starfish, each of the children of $h$ in $T(G)$ is a leaf except for at most one child $h_R$ that represents the head and/or another child representing $2K_1$ or $K_2$. As was seen in \cite{Giakoumakis_OnP4TidyGraphs_1997}, in $\BigO(n_{\pi}(h))$ time, it can be decided whether or not $\pi(h)$ is a starfish (resp.\ an urchin) and, if this is the case, find its partition.

The class of \textit{tree-cographs} \index{tree-cograph} is another superclass of the class of cographs. Tree-cographs were introduced in \cite{TinhoferTreeCographs1988} by the following recursive definition:
\begin{defn}~\ \label{def:treecograph}
\begin{enumerate}
    \item Every tree is a tree-cograph.
    \item If $G$ is a tree-cograph, then $\overline{G}$ is a tree cograph.
    \item The disjoint union of tree cographs is a tree cograph. 
\end{enumerate}
\end{defn}
This definition implies that if $G$ is a tree-cograph, then either $G$ or $\overline{G}$ is disconnected, or $G$ is a tree or the complement of a tree. Hence, both Theorem~\ref{theo:StructureOfP4Tidy} and Definition~\ref{def:treecograph} imply that if $G$ is a tree-cograph and $h$ is an N-node of $T(G)$, then $\pi(h)$ is a tree or the complement of a tree. Note that every tree of more than two nodes is a prime graph, and complements of prime graphs are also prime graphs.

\section{Parameters and Minimal Forbidden Induced Subgraphs} 
\label{sec:parametersMFIS}
In order to effectively use the inherent structure of $P_4$-tidy graphs and tree-cographs, we first explore how the join operation modifies $\alpha_{\mathrm n}$ and $\rho_{\mathrm n}$. As a consequence of these results, we will determine ahead in this section all the minimally non-neighborhood-perfect graphs whose complement is disconnected. As a byproduct, we will also characterize the class of graphs that arises by requiring $\alpha_2=\rho_{\mathrm n}$ for every induced subgraph, which we will call the class of \emph{strongly neighborhood-perfect graph}.

\begin{theorem}\label{teo:PnOfJoin}
If $G$ and $H$ are graphs, then  
\[\pn{G \vee H} = \min \{ \gamma(H) + 1, \gamma(G) + 1, \pn{H},  \pn{G}\}.\]
\end{theorem}

\begin{proof}
It is immediate to see that \[\pn{G\vee H} \leq \min \{ \gamma(H) + 1, \gamma(G) + 1, \pn{H},  \pn{G}\},\] for we can easily find neighborhood sets of $G \vee H$ with all four amounts considered. Simply take a minimum dominating set of either $G$ or $H$ and any vertex in the other graph or, instead, take a minimum neighborhood set in $G$ or $H$.

Let us then prove that indeed the inequality above cannot hold strictly. 

By contradiction let us say that we have a neighborhood set of $S$ of $G\vee H$, with size strictly less than $\min \{ \gamma(H) + 1, \gamma(G) + 1, \pn{H},  \pn{G}\}$. Hence, as $S$ has fewer vertices than $\pn{G}$ and $\pn{H}$, it must have at least one vertex in each $G$ and $H$. For if not, there would be uncovered edges in the subgraphs corresponding to $G$ or $H$ in the join. Thus if we take $S_G = S \cap V(G)$ and $S_H = S \cap V(H)$, then $|S_H| \leq |S|-1$ and $|S_G| \leq |S|-1$. But as we are assuming that $|S|-1 < \gamma(G)$ and $|S|-1 < \gamma(H)$, we have that neither $S_G$ nor $S_H$ can be dominating sets of $G$ and $H$ respectively. This means that there must be at least some $v\in V(G)$ and some $w\in V(H)$ such that $v \notin N_G[S_G]$ and $ w \notin N_H[S_H]$. And then if we take the edge $(v,w)$ in $G\vee H$, it cannot be covered by $S$, for there is no vertex in $S_H$ or $S_G$ adjacent to both vertices and $S = S_G \cup S_H$. Thus, $S$ is not a neighborhood set of the join, reaching the contradiction that proves the theorem. 
\end{proof}

\begin{theorem}\label{teo:AnOfJoin}
If $G$ and $H$ are graphs, then  
\[\an{G \vee H} = \min \{\at{G}, \at{H}\}.\]
\end{theorem}

\begin{proof}
Let us first note that if a neighborhood-independent set of $G\vee H$ has size larger than 1, then it must have no edges belonging to $E(G)$ or $E(H)$. For in the join all edges between vertices of $G$ are in the closed neighborhood of any vertex of $H$ and likewise between the edges of $H$ and the vertices of $G$. Similarly it cannot have any vertices, for every vertex in $G$ is in the closed neighborhood of every vertex of $H$.  

Now, let us prove that $\an{G\vee H} \geq \min \{\at{G}, \at{H}\}$, by finding a neighborhood-independent set of $G\vee H$ of that size. Without loss of generality, suppose $\at{G} \leq \at{H}$. Let $I_G$ be an 2-independent set of $G$ and $I_H$ be one of $H$, both of size $\at{G}$. Clearly as both $I_H$ and $I_G$ are independent sets in $H$ and $G$, then they are also independent sets in $G \vee H$ and so $I_H \cup I_G$ induces a complete bipartite subgraph of $G \vee H$. Let $M$ be a perfect matching between $I_G$ and $I_H$ in $G\vee H$. Clearly $|M|= |I_G| = |I_H| = \at{G}$. We will proceed to show that $M$ is a neighborhood-independent set. 

Suppose to the contrary that there are two edges in $M$, $e_1$ and $e_2$, such that there exists a vertex $u$ of $V(G \vee H)$ satisfying $e_1, e_2 \subseteq N[u]$. Let us write $e_1 = v_1 w_1$ and $e_2 = v_2 w_2$, with $v_1, v_2 \in I_G$ and $w_1, w_2 \in I_H$. As $u$ is a vertex of the join then $u$ must belong to $V(G)$ or $V(H)$. If $u\in V(G)$, then $v_1 u v_2$ is a path of length $2$ from $v_1$ to $v_2$, in $G$. If $u\in V(H)$, then $w_1 u w_2$ is a path of length $2$ in $H$ that connects $w_1$ and $w_2$. In both cases we reach a contradiction, because both $I_G$ and $I_H$ were 2-independent sets. Therefore, $M$ must be a neighborhood-independent set of size $\at{G}$ and the inequality $\an{G\vee H} \geq \min \{\at{G}, \at{H}\}$ must hold.

Now, if $\an{G\vee H}=1$, then by the previous inequality we have the equality we were looking for. Let us then suppose that $\an{G\vee H}>1$, which by the first observation of this proof implies that any neighborhood-independent set of the join must be a matching between vertices of $G$ and $H$. Let $M$ be any neighborhood-independent set of size $\an{G\vee H}$. We define $Y_H$ and $Y_G$ as the sets of vertices of $H$ and $G$ respectively such that $Y_H = \{w \in V(H) \colon \mbox{there exists $e \in M$ such that } w\in e\}$ and $Y_G = \{v \in V(G) \colon \mbox{there exists $e \in M$ such that }v\in e\}$. Clearly $|Y_H| = |Y_G|$, for every edge in $M$ has one vertex in $G$ and one in $H$. We shall see now that both are 2-independent sets. 

Suppose again by contradiction that there are two vertices in $Y_G$, $v_1$ and $v_2$, such that $d_G(v_1,v_2) \leq 2$. This implies that there must exist a vertex $u \in V(G)$ such that $v_1,v_2\in N_G[u]$ which clearly also means that $v_1,v_2\in N_{G\vee H}[u]$. If we now take $w_1$ and $w_2$ in $Y_H$ such that $v_i w_i \in M$ for each $i \in \{1,2\}$, then clearly $v_1 w_1$ and $v_2 w_2$ cannot be neighborhood-independent edges because if $u \in V(G)$, then both $w_1,w_2\in N_{G\vee H}[u]$. This contradicts the fact that $M$ is a neighborhood-independent set. The contradiction proves that $Y_G$ must be a 2-independent set of $G$. By the same reasoning, $Y_H$ must be a 2-independent set of $H$. Hence as $|Y_G| \leq \at{G}$ and $|Y_H| \leq \at{H}$, then $|M| = |Y_H| = |Y_G| \leq \min \{\at{G}, \at{H} \}$ and therefore $\an{G \vee H} \leq \min \{\at{G}, \at{H} \}$, proving the reverse inequality and the theorem.   
\end{proof}

Now we shall extend this results to the join of more than two graphs. For this purpose we state the following lemma, which is easy to prove but still useful.

\begin{lemma}\label{lem:DomSetofMultJoin}
If $G_1,\ldots,G_k$ are graphs and $k\geq 2$, then $$\gamma(G_1\vee\cdots\vee G_k)=\min\{2,\gamma(G_1),\ldots,\gamma(G_k)\}.$$
\end{lemma}
\begin{proof}
Let $G=G_1\vee\cdots\vee G_k$. Clearly, $\gamma(G)\leq\min\{2,\gamma(G_1),\ldots,\gamma(G_k)\}$ since any dominating set of any of the graphs $G_1,\ldots,G_k$ as well as any set $\{v_1,v_2\}$ where $v_1\in V(G_1)$ and $v_2\in V(G_2)$ are dominating sets of $G_1\vee\cdots\vee G_k$. Hence, if the formula were false, then $\gamma(G)<\min\{2,\gamma(G_1),\ldots,\gamma(G_k)\}$, which means that $\gamma(G)=1$ and $\gamma(G_1),\ldots,\gamma(G_k)$ are greater than $1$ all of them. Therefore, $G$ has a universal vertex but none of $G_1,\ldots,G_k$ has a universal vertex, which contradiction the fact that $G=G_1\vee\cdots\vee G_k$.
\end{proof}

We now give a formula of the neighborhood number for the join of more than two graphs.
\begin{corollary}\label{cor:PnOfMultJoin}
If $G_1,G_2,\ldots,G_k$ are graphs and $k\geq 3$, then
\[ \pn{G_1\vee\cdots\vee G_k} = \min \{3,\gamma(G_1)+1,\cdots,\gamma(G_k)+1,\pn{G_1},\cdots,\pn{G_k}\}.\]
\end{corollary}
\begin{proof}    
The formula is valid when $k=3$ because Theorem~\ref{teo:PnOfJoin} and Theorem~\ref{lem:DomSetofMultJoin} imply
\begin{align*}
   \rho_{\mathrm n}(G_1\vee G_2\vee G_3)
      &=\min\{\gamma(G_1\vee G_2)+1,\gamma(G_3)+1,\rho_{\mathrm n}(G_1\vee G_2),\rho_{\mathrm n}(G_3)\}\\
      &=\min\{\min\{2,\gamma(G_1),\gamma(G_2)\}+1,\gamma(G_3)+1,\\
        &\phantom{=\min\{}\min\{\gamma(G_1)+1,\gamma(G_2)+1,\rho_{\mathrm n}(G_1),\rho_{\mathrm n}(G_2)\},\rho_{\mathrm n}(G_3)\}\\
      &=\min\{3,\gamma(G_1)+1,\gamma(G_2)+1,\gamma(G_3)+1,\rho_{\mathrm n}(G_1),\rho_{\mathrm n}(G_2)\},\rho_{\mathrm n}(G_3)\}
\end{align*}
Moreover, if the formula is valid when $k=t$ for some $t\geq 3$, then it is also valid when $k=t+1$ since Lemma~\ref{lem:DomSetofMultJoin} implies
\begin{align*}
   \rho_{\mathrm n}(G_1\vee\cdots\vee G_{t+1})
      &=\min\{\gamma(G_1\vee\cdots\vee G_t)+1,\gamma(G_{t+1})+1,\\
      &\phantom{=\min\{}\rho_{\mathrm n}(G_1\vee\cdots\vee G_t),\rho_{\mathrm n}(G_{t+1})\}\\
      &=\min\{\min\{2,\gamma(G_1),\cdots,\gamma(G_t)\}+1,\gamma(G_{t+1})+1,\\
      &\phantom{=\min\{}\min\{3,\gamma(G_1)+1,\cdots,\gamma(G_t),\rho_{\mathrm n}(G_1),\cdots,\rho_{\mathrm n}(G_t)\},\\
       &\phantom{=\min\{}\rho_{\mathrm n}(G_{t+1})\}\\
      &=\min\{3,\gamma(G_1)+1,\cdots,\gamma(G_{t+1})+1,\rho_{\mathrm n}(G_1),\ldots,\rho_{\mathrm n}(G_{t+1})\}.
\end{align*}
By induction, the formula is valid for every $k\geq 3$.
\end{proof}

And now we state the following immediate consequence of Theorem~\ref{teo:AnOfJoin}, for future reference.

\begin{corollary}\label{cor:AnOfMultJoin}
If $G_1,\ldots,G_k$ are graphs and $k\geq 3$, then
\[ \alpha_{\mathrm n}(G_1\vee\cdots\vee G_k)=1. \]    
\end{corollary}
\begin{proof}
Since every two vertices of $G_1\vee\cdots\vee G_{k-1}$ are at distance at most two, $\alpha_2(G_1\vee\cdots\vee G_{k-1})=1$. Hence, Theorem~\ref{teo:AnOfJoin} implies that $\alpha_{\mathrm n}(G_1\vee\cdots\vee G_k)=\min\{\alpha_2(G_1\vee\cdots\vee G_{k-1}),\alpha_2(G_k)\}=1$.
\end{proof}

 Using the previous two results, we will prove the following theorem, which states which graphs that are formed by the join of two non-null subgraphs are minimally non-neighborhood-perfect or equivalently which are the only minimally non-neighbor\-hood-perfect graphs that have a disconnected complement.

\begin{theorem}\label{theo:CaractOfMinimallyNonNPwithDiscComp}
The only minimally non-neighborhood-perfect graphs with disconnected complement are $C_6 \vee 3K_1$, $P_6 \vee 3K_1$, and $C_4 \vee 2K_1 = \overline{3K_2}$.    
\end{theorem}

For the purpose of proving Theorem~\ref{theo:CaractOfMinimallyNonNPwithDiscComp} we shall first define a subclass of neighborhood-perfect graphs, the strongly neighborhood-perfect graphs.

\begin{defn}
We shall say that a graph $G$, is \emph{strongly neighborhood-perfect} if $\at{G'} = \pn{G'}$ for every induced subgraph $G'$ of $G$. 
\end{defn}

\begin{defn}
We shall say that a graph $G$, is \emph{minimally non-strongly neigh\-bor\-hood-perfect} when $\at{G} < \pn{G}$, but $\at{G'} = \pn{G'}$ for every proper induced subgraph $G'$, of $G$. That is, it is not strongly neighborhood-perfect, but all its proper induced subgraphs are.   
\end{defn} 

\begin{rem}\label{obs:StronglyNpAreNp}
Clearly all strongly neighborhood-perfect graphs are neighbor\-hood-perfect. It follows from the string of inequalities: $\at{G} \leq \gamma(G) \leq \an{G} \leq \pn{G}$, that holds for every graph $G$, and the equality demanded by the definition of strongly neighborhood-perfect graphs. Moreover it is also true that if $G$ is neighborhood-perfect, then it is strongly neighborhood-perfect if and only if $\at{G'} = \an{G'}$ for every induced subgraph $G'$. 
\end{rem}

We shall see which graphs satisfy that $\at{G'} = \an{G'}$ for every induced subgraph $G'$. But before giving this characterization we shall prove a useful general property of chordal $P_k$-free graphs.

\begin{lemma}\label{lemma:kWalkinPkfreeChordal}
Any $k$-walk $W$ in a $P_k$-free chordal graph must have at least two vertices that are $2$ steps from each other in $W$ and either are adjacent or the same vertex.   
\end{lemma}

\begin{proof}
Let $G$ be a $P_k$-free chordal graph and $W$ be a $k$-walk in $G$. Since $G$ is $P_k$-free, then $W$ cannot be an induced path. This means that there must exist an integer $p$ such that $p \geq 2$ and there are at least two vertices of $W$ which are $p$ steps from each other in $W$ and are either adjacent in $G$ or the same vertex in $G$. We choose $p$ as small as possible. If we show that $p=2$, then the assertion of the lemma follows.

Let us suppose by contradiction that $p$ is greater that $2$. We take two vertices in $W$ that are $p$ steps from each other and either are adjacent or the same vertex in $G$, and consider the sub-walk of $W$ of length $p$ joining them. As the minimality of $p$ implies that the vertices that are at fewer than $p$ steps from each other in $W$ are different and nonadjacent, this sub-walk must induce $C_p$ or $C_{p+1}$ in $G$, depending on whether the two vertices are the same or adjacent. But as the graph was chordal and $p$ was greater than $2$, this results in a contradiction, proving the lemma.
\end{proof}

\begin{lemma}\label{lemma:a2EqualAniffP6-FreeChordal}
A graph $G$ satisfies $\at{G'} = \an{G'}$ for every induced subgraph $G'$ of $G$ if and only if $G$ is $P_6$-free chordal. 
\end{lemma}

\begin{proof}
Let $S \subseteq V(G) \cup E(G)$ be a neighborhood-independent set of size $\an{G}$ and of minimum number of edges. We shall show that $S$ must contain only vertices and therefore be a 2-stable set of $G$, proving the lemma (for $\at{G} \leq \an{G}$ is clearly true for all graphs). 

Assume to the contrary that there is an edge $e = x y \in S$. As $e$ cannot be replaced by $x$ in $S$, maintaining the neighborhood-independence (for $S$ had minimum number of edges), then there must exist an $s\in S$ (an edge or vertex), such that $N[x] \cap N[s] \neq \emptyset$. But as $e, s\in S$, then $N[x]$ and $N[y]$ and $N[s] = \emptyset$, which means that there is a vertex $x'\in N[x] \cap N[s]$ such that $x' \notin N[y]$. Moreover $x\in N[x] \cap N[y]$, which implies that $x\notin N[s]$, meaning that there must be a vertex $x''$, such that $x'' \notin N[x]$ and either $x''\in s$ if $s$ is an edge or $x''=s$ if $s$ is a vertex. But as $x'' \in s$ (or $x''=s$), and $x' \in N[s]$, then $x'' \in N[x']$ and $x'' \notin N[x]$. By a symmetry argument, there must be vertices $y'$ and $y''$, such that $y'\in N[y] - N[x]$ and $y'' \in N[y'] - N[y]$. But then $x''$ $x'$ $x$ $y$ $y'$ $y''$ form a 6-walk where no two vertices that are two steps from each other are adjacent or the same. This together with Lemma~\ref{lemma:kWalkinPkfreeChordal}, results in a contradiction, proving that no edge can belong to $S$ and therefore $S$ must be a 2-independent set of size $\an{G}$. 
\end{proof}

Using the previous characterization, we shall state the following corollary, fully characterizing strongly neighborhood-perfect graphs by forbidden induced subgraphs.

\begin{corollary}\label{cor:caractSNP}
If $G$ is a graph, the following statements are equivalent:

\begin{enumerate}
\item $G$ is strongly neighborhood-perfect
\item $G$ is neighborhood-perfect $\cap$ $\{ C_4, C_6, P_6 \}$-free
\item $G$ is odd-sun-free $\cap$ $P_6$-free chordal
\end{enumerate}
\end{corollary} 

\begin{proof}
Clearly, by Remark~\ref{obs:StronglyNpAreNp}, $G$ is strongly neighborhood-perfect if and only if $G$ is neighborhood-perfect and $\at{G'} = \an{G'}$ for every induced subgraph $G'$, which, by Lemma~\ref{lemma:a2EqualAniffP6-FreeChordal}, holds if and only if $G$ is neighbor\-hood-perfect and $P_6$-free chordal. But, as all odd holes are forbidden induced subgraphs of neighborhood-perfect graphs \cite{lehel_neighborhood_1986}, then $G$ is neighbor\-hood-perfect and $P_6$-free chordal if and only if it is neighborhood-perfect and $\{C_4, C_6, P_6\}$-free, proving (1) if and only if (2). Moreover, by Theorem~\ref{thm:np-chordal-forb} a chordal graph is neighborhood-perfect if and only if it is odd-sun-free, clearly implying (2) if and only if (3).
\end{proof}

We shall now give the proof of Theorem~\ref{theo:CaractOfMinimallyNonNPwithDiscComp}, that states that the only minimally non-neighborhood-perfect graphs with disconnected complement (which means that they are join of two non-null subgraphs) are $C_4 \vee 2K_1$, $C_6 \vee 3K_1$ and $P_6 \vee 3K_1$. 

\begin{proof}[Proof of Theorem~\ref{theo:CaractOfMinimallyNonNPwithDiscComp}]

Clearly a graph with disconnected complement can be thought of as the join of two non-null graphs. Let us consider we have a minimally non-neighbor\-hood-perfect graph $G \vee H$. By minimality $G$ and $H$ must be neighborhood-perfect, but as {$G\vee H$} is minimally non-neighborhood-perfect, $\pn{G \vee H} \neq \an{ G \vee H } $ which, by Theorem~\ref{teo:PnOfJoin} and Theorem~\ref{teo:AnOfJoin}, implies that $G$ or $H$ must satisfy $\alpha_2 \neq \rho_{\mathrm n}$ (because $\at{W} < \gamma(W)+1$ is true for all graphs $W$). 

We shall note that if a graph is neighborhood-perfect but does not satisfy $\alpha_2 = \rho_{\mathrm n}$, then it is neighborhood-perfect but not \textbf{strongly} neighborhood-perfect, which, by (1) if and only if (2) in Corollary~\ref{cor:caractSNP}, means that the graph must contain a $C_4$, $C_6$ or $P_6$ as induced subgraph.

Let us then suppose that both $G$ and $H$ do not satisfy $\alpha_{\mathrm n}=\rho_{\mathrm n}$. This means that both must contain a $C_4$, $C_6$ or $P_6$ as induced subgraph. If one of them contains an induced $C_4$, then $G \vee H$ must have $C_4 \vee 2K_1 = \overline{3K_2}$ as a proper induced subgraph, contradicting the minimality of $G\vee H$. On the other hand if none contain an induced $C_4$, then they must contain an induced $C_6$ or $P_6$, meaning that both must have at least an independent set of size 3. Hence $G \vee H$ must contain a $P_6 \vee 3K_1$ or $C_6 \vee 3K_1$ as a proper induced subgraph, but by Theorem~\ref{teo:PnOfJoin} and Theorem~\ref{teo:AnOfJoin}, both have $\rho_{\mathrm n} = 3 \neq 2 = \alpha_{\mathrm n}$, meaning that they are not neighborhood-perfect. In both cases we have found a contradiction, therefore it cannot occur that both $G$ and $H$ do not satisfy $\alpha_{\mathrm n} = \rho_{\mathrm n}$.

We need only to consider the case where one of the graphs does not satisfy $\alpha_2 = \rho_{\mathrm n}$. Let us say that $G$ does not satisfy $\at(G) =\pn(G)$ and, consequently, $G$ has an induced subgraph $G'$ isomorphic to $C_4$, $C_6$ or $P_6$. Now, as only $G$ does not satisfy $\at{G} = \pn{G}$, then $\at{G} < \pn{G}$ and $\at{H} = \pn{H}$. Hence $\at{G} < \at{H} \leq \alpha(H)$ because, if not $G\vee H$ would be neighborhood-perfect. Thus we can take $H' = (\at{G'}+1)K_1$ as an induced subgraph of $H$, for $\at{G'} \leq \at{G} < \alpha(H)$. Then once again $G' \vee H'$ must be a $\overline{3K_2}$, $C_6 \vee 3K_1$ or $P_6 \vee 3K_1$. But now as $G' \vee H'$ is an induced subgraph of $G \vee H$, by minimality $G \vee H$ = $G' \vee H'$, proving the theorem.\end{proof}

\section{Structural Characterizations}\label{sec:structural_characterizations}

In this section we shall characterize by minimal forbidden induced subgraphs the class of neighborhood-perfect graphs, restricted to the classes of $P_4$-tidy graphs and tree-cographs. For this we will strongly rely on the characterization of minimally non-neighborhood-perfect graphs with disconnected complement shown in the previous section.

Using Theorem~\ref{theo:CaractOfMinimallyNonNPwithDiscComp} together with the fact that every disjoint union of neighbor\-hood-perfect graphs is neighborhood-perfect and the structures of $P_4$-tidy and tree-cographs given by Theorem~\ref{theo:StructureOfP4Tidy} and Definition~\ref{def:treecograph}, respectively, the following two forbidden induced subgraph characterizations can be proven.

\begin{theorem} \label{teo:CaractOfP4TidyNP}
If $G$ is a $P_4$-tidy graph, then it is neighborhood-perfect if and only if it is $\{\overline{3K_2}, 3-$sun$, C_5 \}$-free.
\end{theorem}

\begin{theorem}\label{teo:CaractNPTreeCograph}
If $G$ is a tree-cograph, then $G$ is neighborhood-perfect if and only if $G$ is $\{\overline{3K_2}, P_6 \vee 3K_1 \}$-free.
\end{theorem}

We shall first prove Theorem~\ref{teo:CaractOfP4TidyNP}. Let us then begin by determining the values of $\an{G}$ and $\pn{G}$ for any connected and co-connected $P_4$-tidy graph $G$.

\begin{theorem} \label{teo:ConnectedAndCoConnectedP4TidyGraphs}
If $G$ is a nontrivial connected and co-connected $P_4$-tidy graph, then one of the following statements holds:
\begin{enumerate}
    \item $G$ is isomorphic to $C_5$, $\pn{G} = 3$ and $\an{G}= 2$.
    \item $G$ is isomorphic to $P_5$ or $\overline{P_5}$ and $\an{G} = \pn{G} = 2$.
    \item $G$ is a starfish with $t$ ends or a fat starfish arising from one, and \\$\an{G} = \pn{G} = t$.
    \item $G$ is an urchin or a fat urchin with at least $3$ ends, and \\ $\pn{G}= 2$, $\an{G}=1$.
\end{enumerate}    
\end{theorem}  
\begin{proof}
Since $G$ is $P_4$-tidy, connected and co-connected, it follows by Theorem~\ref{theo:StructureOfP4Tidy} that $G$ is isomorphic to $C_5$, $P_5$, $\overline{P_5}$, a starfish, a fat starfish, an urchin or a fat urchin. The values of $\rho_{\mathrm n}$ and $\alpha_{\mathrm n}$ for $C_5$, $P_5$ and $\overline{P_5}$ can be easily checked by simple inspection. 

We shall then consider first the case where $G$ is a starfish with partition $(S,C,R)$, such that $|S|=t$, or a fat starfish arising from such a starfish by the substitution of a vertex $c$ of $C$ by a $K_2$, or by the substitution of a vertex $s$ from $S$ by a $K_2$ or $2K_1$. In all cases there is a neighborhood set of size $t$ formed by taking $t$ vertices from $C$. If $G$ is a starfish without substitution of a vertex of $C$ then we take all $C$, if on the other hand it is a starfish where a vertex $c$ of $C$ has been substituted by a $K_2$ or $2K_1$, we take only one of the vertices by which $c$ has been substituted. If $G$ is a fat starfish arising by substituting a vertex $c$ of $C$ by a $2K_1$, then $C-\{c\} \cup \{s\}$, where $s$ was the only neighbor of $c$ in $S$, is a neighborhood set of size $t$ of $G$. Thus $\pn{G} \leq t$. Now in all previous cases, if we take $t$ edges that connect $S$ to $C$, we get a neighborhood-independent set of size $t$. In the cases where a vertex has been substituted by a $K_2$ or $2K_1$, we choose only one of the two edges from $S$ to $C$ involved and all the other edges from $S$ to $C$. In the case of a starfish that is not fat, we take all edges from $S$ to $C$. Thus we have found in all cases a neighborhood-independent set of size $t$, implying that $\an{G} \geq t$. And as $\an{G} \leq \pn{G}$, we have that $\an{G} = \pn{G} = t$. 

Let us now note that an urchin (or fat urchin) of less than $3$ ends is also a starfish (or fat starfish). Therefore if we assume without loss of generality that $G$ is not a starfish, the only possibility remaining is that $G$ is an urchin with at least $3$ ends. 

If $G$ is an urchin or fat urchin with partition $(S,C,R)$, and $|S| = t \geq 3$, we shall see that $\an{G} = 1$ and $\pn{G} = 2$. As there is no universal vertex in $G$, then $\pn{G} \geq 2$. Moreover, if we take two vertices of $C$, taking care of not taking any vertex from the substituting $K_2$ or $2K_1$ in case $G$ is a fat urchin, we clearly obtain a neighborhood set. Hence clearly $\pn{G} = 2$. Now let us see that indeed we cannot have a neighborhood independent set of size $2$. This becomes clear if we observe that in all cases, if $G$ is an urchin or a fat urchin, all vertices and edges are in at least the neighborhood of $t-1$ vertices of $C$. That is, except for the vertices in $S$ (or, eventually, of the $K_2$ or $2K_1$ substituting a vertex of $S$), all the rest of the vertices are adjacent to all vertices in $C$, and these are adjacent to $t-1$ vertices of $C$. Moreover all edges between vertices of $R \cup C$ are in the neighborhood of $t$ vertices of $C$, and all edges between vertices of $S$ and $C$ are in the neighborhood of $t-1$ vertices of $C$. Thus, if we take any two edges or vertices of $G$, as $t \geq 3$, then there must at least be one vertex of $C$ that includes them both in its neighborhood. Therefore, $\an{G} = 1$.  
\end{proof}

We shall then proceed to prove Theorem~\ref{teo:CaractOfP4TidyNP}.

\begin{proof}[Proof of Theorem~\ref{teo:CaractOfP4TidyNP}]

If $G$ is neighborhood-perfect, then it cannot contain as induced subgraph a $\overline{3K_2}$, $3$-sun or $C_5$ because none of these graphs are neigh\-bor\-hood-perfect and the class of neighborhood-perfect graphs is hereditary. We must then only prove that if $G$ is not neighborhood-perfect then it must contain $\overline{3K_2}$, $3$-sun, or $C_5$ as an induced subgraph. 

Suppose that $G$ is a $P_4$-tidy graph which is not neighborhood-perfect. Then it must contain a minimally non-neighborhood-perfect graph as induced subgraph; let $H$ be any such subgraph. The minimality of $H$ implies that it must be connected. If $\overline{H}$ is disconnected, then $H$ is a minimally non-neighborhood-perfect graph with disconnected complement, which by Theorem~\ref{theo:CaractOfMinimallyNonNPwithDiscComp} means that it must be $C_4 \vee 2K_1 = \overline{3K_2}$ , $C_6 \vee 3K_1$ or $P_6 \vee 3K_1$.  But as the class of $P_4$-tidy graphs is hereditary, $H$ must be $P_4$-tidy, which implies that it cannot be $C_6 \vee 3K_1$ or $P_6 \vee 3K_1$. This is because both graphs contain four vertices with at least two companion vertices, namely any consecutive four vertices of the $C_6$ or the center vertices of the $P_6$, and therefore are not $P_4$-tidy. Hence if $\overline{H}$ is disconnected, then $H$ can only be $\overline{3K_2}$. 

Let us suppose now that both $H$ and $\overline{H}$ are connected. As $H$ is minimally non-neighborhood-perfect, then $\an{H}$ must be different from $\pn{H}$. Which means, by Theorem~\ref{teo:ConnectedAndCoConnectedP4TidyGraphs}, that $H$ must be a $C_5$ or an urchin or fat urchin with at least 3 ends. Lastly, if $H$ is an urchin or fat urchin with at least 3 ends, then it must have $\an{H}=1$ and $\pn{H}=2$. But by Theorem~\ref{teo:MinNonNPwithAN1Are3PyrAnd0Pyr}, the only minimally non-neighborhood-perfect graphs with $\an{H}=1$ are the $\overline{3K_2}$ and $3$-sun and the only one of these that is an urchin is the $3$-sun. Therefore, as $H$ is connected and co-connected, it must be a $C_5$ or a $3$-sun. 

We conclude that $H$ must be isomorphic to $\overline{3K_2}$, $C_5$ or $3$-sun and since, by construction, $H$ is an induced subgraph of $G$, this proves the theorem.\end{proof}

Having proved Theorem~\ref{teo:CaractOfP4TidyNP}, we proceed to prove the characterization of neighborhood-perfect graphs restricted to the class of tree-cographs. We shall work with the structural definition of a tree-cograph and strongly rely on the characterization of minimally non-neighborhood-perfect graphs with disconnected complement given in Theorem~\ref{theo:CaractOfMinimallyNonNPwithDiscComp}. 

\begin{theorem} \label{teo:ConCoConTreeCograph}
    If $G$ is a connected and co-connected tree-cograph, then one of the following statements holds:
    \begin{enumerate}
             \item{ $G$ is a tree and $\pn{G} = \an{G} = \nu(G) = \tau(G)$,}
             \item{ $G$ is a connected co-tree and $\pn{G}=2$.}
    \end{enumerate}     
\end{theorem} 

\begin{proof}
By the definition of tree-cographs, if $G$ is connected and co-connected, then $G$ must be a tree with connected complement or a connected co-tree. 

If $G$ is a tree then it is bipartite. It was already noted in \cite{lehel_neighborhood_1986} and \cite{Sampathkumar_Neeralagi_NNumber_1985} that for any bipartite graph $G$, $\an{G} = \nu(G)$ and $\pn{G} = \tau(G)$, which by the K\"onig-Egerv\'ary theorem implies that $\an{G} = \nu(G) = \tau(G) = \pn{G} $.

If $G$ is a connected co-tree, then $\overline{G}$ has at least one leaf; that leaf and its only neighbor in $\overline{G}$ clearly form a neighborhood set of $G$ of size $2$. Moreover as $G$ has connected complement, there cannot be a neighborhood set of size $1$, for this would imply the existence of a universal vertex in $G$ and an isolated vertex in $\overline{G}$. Hence $\pn{G}=2$, proving the theorem. 
\end{proof}

\begin{corollary} \label{cor:ThereAreNoConCoConTreeCogMinNonNP}
    There are no connected and co-connected tree-cographs that are minimally non-neighborhood-perfect. 
\end{corollary}

\begin{proof}
If a graph $G$ is a connected and co-connected tree-cograph, then, by definition, $G$ is a tree or a co-tree. By Theorem~\ref{teo:ConCoConTreeCograph}, a tree cannot be non-neighborhood-perfect. Moreoveer, if $G$ is a co-tree, then $\pn{G}=2$, which means that if $G$ is minimally non-neighborhood-perfect, then $\an{G}$ must be $1$. But by Theorem~\ref{teo:MinNonNPwithAN1Are3PyrAnd0Pyr}, the only minimally non-neighborhood-perfect graphs with $\an{G}=1$ are the $3$-sun and the $\overline{3K_2}$, none of which are co-trees. Hence if $G$ is a minimally non-neighborhood-perfect graph, $G$ cannot be a connected and co-connected tree-cograph.   
\end{proof}

Having proved this previous results, we give the proof of Theorem~\ref{teo:CaractNPTreeCograph}, that states that if $G$ is a tree-cograph, then $G$ is neighborhood-perfect if and only if $G$ is $\{\overline{3K_2}, P_6 \vee 3K_1 \}$-free.

\begin{proof}[Proof of Theorem~\ref{teo:CaractNPTreeCograph}]
It is clear that if $G$ is neighborhood-perfect, it cannot have $\overline{3K_2}$ or $P_6 \vee 3K_1$ as induced subgraphs, for they are both minimally non-neighborhood-perfect graphs. We shall now prove that if it does not have those graphs as subgraphs, then it is neighborhood-perfect.

Suppose that $G$ is not neighborhood-perfect. Hence $G$ must contain an induced subgraph $H$ that is minimally non-neighborhood-perfect. Clearly by minimality, $H$ cannot be disconnected. If $H$ has disconnected complement, then it is a minimally non-neighborhood-perfect graph with disconnected complement, and by Theorem~\ref{theo:CaractOfMinimallyNonNPwithDiscComp}, it must be $C_4 \vee 2K_1 = \overline{3K_2}$, $C_6 \vee 3K_1$ or $P_6 \vee 3K_1$. But $C_6 \vee 3K_1$ is not a tree-cograph, because it is clearly neither a tree, nor a co-tree, nor the disjoint union of two tree-cographs nor the join of two tree-cographs. Thus if $H$ has disconnected complement, it must be  $\overline{3K_2}$ or $P_6 \vee 3K_1$. 

On the other hand if $H$ has a connected complement, then it will be a connected and co-connected tree-cograph. However by Corollary~\ref{cor:ThereAreNoConCoConTreeCogMinNonNP}, if $H$ is minimally non-neighborhood-perfect, then it cannot be a connected and co-connected tree-cograph. Hence $H$ can only be $\overline{3K_2}$ or $P_6 \vee 3K_1$, and as $H$ was by construction an induced subgraph of $G$, this proves the theorem.  
\end{proof}

\section{Algorithms and Complexity Results} 
\label{sec:algorithms_and_complexity_results}

In this section we shall present linear-time algorithms to solve the recognition problems of neighborhood-perfect graphs, as well the problems of finding an optimal neighborhood-independent set and neighborhood-covering set, restricted to the classes of $P_4$-tidy graphs and tree-cographs. Moreover, we shall prove that although the problem of determining $\an(G)$ and $\pn(G)$ can be solved in linear time for complements of trees, it becomes $\mathcal{NP}$-hard for complements of bipartite graphs.

\subsection{Recognition Algorithms}

By using the particular structure of the modular decomposition trees of $P_4$-tidy graphs and tree-cographs, and Theorems~\ref{teo:CaractOfP4TidyNP} and \ref{teo:CaractNPTreeCograph}, we will show two linear-time algorithms that solve the recognition problem of neighborhood-perfect graphs restricted to these two classes. Both of these algorithms work on the modular decomposition tree of the input graph.

We shall begin by describing a linear-time algorithm that decides whether or not a $P_4$-tidy graph is neighborhood-perfect. Let us first remember that, as was said in section~\ref{sec:preliminaries}, it follows from Theorem~\ref{theo:StructureOfP4Tidy} that if $h$ is an N-node of the modular decomposition tree of a $P_4$-tidy graph $G$, then $\pi(h)$ must be isomorphic to $C_5$, $P_5$, $\overline{P_5}$, a prime starfish or a prime urchin. Moreover in $\BigO(n_{\pi}(h))$ time it can be decided whether or not $\pi(h)$ is a starfish (resp.\ urchin) and, if affirmative, its partition can be found within the same time bound.

Our recognition algorithm for neighborhood-perfect graphs, restricted to the class of $P_4$-tidy graphs, performs a simple traversal of the modular decomposition tree of the input graph, which, we shall show, makes the algorithm terminate in $\BigO(n)$ time provided the modular decomposition tree is given as an input. The algorithm will strongly rely on the characterization by forbidden induced subgraphs proven in Theorem~\ref{teo:CaractOfP4TidyNP}. 

In order to simplify the recognition algorithm, we shall first define a boolean function $C: V(T(G)) \to \{True, False\}$, where $T(G)$ is the modular decomposition tree of $G$, and, for each node $h$ of $T(G)$, $C(h) = True$ if and only if $G[h]$ contains an induced $C_4$. We shall prove that, given as input the modular decomposition tree $T(G)$ of any $P_4$-tidy graph $G$, Algorithm~\ref{algo:C4Test} can be implemented so as to compute $C(h)$ for each node $h$ of $T(G)$ in $\mathcal O(n)$ overall time.  Once we have proved so, we shall use Algorithm~\ref{algo:C4Test} as a subroutine in Algorithm~\ref{algo:RecP4TidyNP}, which recognizes neighborhood-perfect graphs in the class of $P_4$-tidy graphs.

\begin{algorithm2e}[ht!]
\small
\DontPrintSemicolon
\SetKwInput{Input}{Input}
\SetKwInput{Output}{Output}
\SetKwBlock{StepOne}{Step 1: \; Traverse the nodes of $T(G)$ in post-order, and in each node $h$ \textbf{do}:}{end}
\SetKwBlock{StepTwo}{Step 2:}{end}

\Input{A $P_4$-tidy graph $G$ and its modular decomposition tree $T(G)$}
\Output{The modular decomposition tree $T(G)$ with the value $C(h)$ attached to each node $h$ of it, where $C(h) = True$ if and only if $G[h]$ contains an induced $C_4$}    

\StepOne{   
    \lIf{$h$ is a leaf}{$C(h) := False$} 
    \ElseIf{$C(h')$ is True for any child $h'$ of $h$ \textbf{\textup{or}}\\
            $h$ is a P-node having at least two nonleaf children \textbf{\textup{or}}\\
            $\pi(h)$ is $\overline{P_5}$ \textbf{\textup{or}}\\
            $\pi(h)$ is a starfish or an urchin and any vertex of its body represents $2K_1$}
            {$C(h) := True$}
    \lElse{$C(h) := False$}
}
\StepTwo{Output $C(h)$ for every node $h$ of $T(G)$}

\caption{Computes $C(h)$ for every node of $T(G)$, with $G$ a $P_4$-tidy graph \label{algo:C4Test}}
\end{algorithm2e}

\begin{algorithm2e}[ht!]
\small
\DontPrintSemicolon
\SetKwInput{Input}{Input}
\SetKwInput{Initialization}{Initialization}
\SetKwInput{Output}{Output}
\SetKwBlock{StepOne}{Step 1: \; Traverse every node $h$ of $T(G)$ in any order and \textbf{do}:}{end}
\SetKwBlock{StepTwo}{Step 2:}{end}

    \Input{A $P_4$-tidy graph $G$}
    \Output{Determines whether or not $G$ is a neighborhood-perfect graph}
    \Initialization{Build the modular decomposition tree $T(G)$ of $G$ and compute $C(h)$ for every node $h$ of $T(G)$ using Algorithm~\ref{algo:C4Test} }
    
\StepOne{
        \If{$h$ is an N-node}{
            \If{$\pi(h)$ is a $C_5$ or an urchin with at least 3 ends}{
                \textbf{output} ``$G$ is not neighborhood-perfect'' and \textbf{stop}
                \label{ln:2C5}
            }
            \ElseIf{$\pi(h)$ is a fat starfish such that a vertex of its body represents $2K_1$ and $C(h_r)$ is True where $h_r$ is the only vertex of its head}{
                \lIf{$C(h_r)$ is True for $h_r$ the child representing the head of $\pi(h)$}{ \textbf{output} ``$G$ is not neighborhood-perfect'' and \textbf{stop} \label{ln:2Starfish}}
            }
        }

        \ElseIf{$h$ is an S-node}{
            \If{$h$ has at least three nonleaf children}{
                \textbf{output} ``$G$ is not neighborhood-perfect'' and \textbf{stop} \label{ln:2SnodeThreeChildren}
            }
            \ElseIf{$h$ has exactly two nonleaf children $h_1$ and $h_2$ and at least one of $C(h_1)$ and $C(h_2)$ is True}{\textbf{output} ``$G$ is not neighborhood-perfect'' and \textbf{stop} \label{ln:2SnodeTwoChildren}}
        }
    
}
\StepTwo{ \textbf{output} ``$G$ is neighborhood-perfect''}
\caption{Recognition of neighborhood-perfectness of $P_4$-tidy\label{algo:RecP4TidyNP}}
\end{algorithm2e}

Below, we prove that these two algorithms are indeed correct and run in linear-time.

\begin{theorem} \label{teo:C4TestCorrectness}
    Algorithm \ref{algo:C4Test} correctly computes $C(h)$ for every node $h$ of any given modular decomposition tree $T(G)$ in $\BigO(n)$ time, whenever $G$ is a $P_4$-tidy graph.
\end{theorem}
\begin{proof}
   Clearly the algorithm sets $C(h)$ correctly for each leaf $h$ of $T(G)$. Let $h$ be any nonleaf node of $T(G)$ and suppose, without loss of generality, that the algorithm correctly sets $C(h')$ for each of the nodes $h'$ visited before $h$. It is then easy to check that if the algorithm sets $C(h)$ to True, $G[h]$ contains an induced $C_4$. Conversely, suppose that $G[h]$ contains an induced $C_4$ and we shall prove that the algorithm correctly sets $C(h)$ to True. Thus, if any vertex $h'$ of $\pi(h)$ represents a graph containing an induced $C_4$, then $C(h')$ is set to True and consequently also $C(h)$ is set to True. Hence, we assume without loss of generality, that every vertex of $\pi(h)$ represents a $C_4$-free graph. Since $G[h]$ contains an induced $C_4$, Theorem 2.1 implies that $\pi(h)$ is $\overline{P_5}$, an edgeless graph, a complete graph, a starfish or an urchin. If $\pi(h)$ contains an induced $C_4$, necessarily $\pi(h)$ is isomorphic to $\overline{P_5}$ and the algorithm correctly sets $C(h)$ to True. Thus, we assume without loss of generality, that $\pi(h)$ is $C_4$-free. In particular, $G[h]$ is not $\overline{P_5}$. If there is a nonsimplicial vertex $h'$ of $\pi(h)$ representing a non-complete graph, then $G[h]$ is a starfish or an urchin and $h'$ is a vertex of the body representing a non-complete graph; if so, Theorem~\ref{theo:StructureOfP4Tidy} implies that $h'$ represents $2K_1$ and the algorithm correctly sets $C(h)$ to True. Hence, we assume without loss of generality, that every nonsimplicial vertex of $\pi(h)$ represents a complete graph. We conclude that each induced $C_4$ of $G[h]$ arises from two adjacent simplicial vertices $h_1$ and $h_2$ of $\pi(h)$, each of which represents a non-complete graph. Necessarily, $h$ is a $P$-node and $h_1$ and $h_2$ are nonleafs. Also in this case the algorithm correctly sets $C(h)$ to True. This completes the proof of the correctness of the algorithm.

   As for the complexity of the algorithm, it is clear that each node is seen only once, and that every node is traversed after all of its children. Hence, as $T(G)$ has at most $2n$ nodes, the algorithm can easily be implemented to check for every node $h$ if $C(h')$ is True for some child $h'$ or if $h$ is a P-node with at least two nonleaf children, all in $\BigO(n)$ time. Moreover, $\pi(h)$ is $\overline{P_5}$, an urchin or starfish only if $h$ is an N-node. In $\BigO(n_{\pi}(h))$ time it can be checked if any N-node $h$ of a $P_4$-tidy graph is $P_5$, $C_5$, $\overline{P_5}$ or an urchin or starfish and in these cases find their partitions. Thus, it can be verified for all N-nodes $h$ if $\pi(h)$ is $\overline{P_5}$ or if it is a starfish or urchin with a vertex of its body representing a $2K_1$, in time $\BigO(\sum_{h \text{ an N-node}} n_{\pi}(h))$. As was already stated in Section~\ref{sec:preliminaries}, the sum of $n_{\pi}(h)$ for all N-nodes $h$ is at most $2n$. Therefore the whole algorithm can be implemented in $\BigO(n)$ time.
\end{proof} 

\begin{theorem} \label{teo:CorRecP4TidyNP}
    Algorithm~\ref{algo:RecP4TidyNP} correctly determines if a $P_4$-tidy graph $G$ is neigh\-bor\-hood-perfect, in linear-time. Moreover, it works in $\BigO(n)$ time if the modular decomposition tree of $G$ is given as part of the input.
\end{theorem}
\begin{proof}
    In order to prove that Algorithm~\ref{algo:RecP4TidyNP} correctly decides neighborhood-perfectness of any given $P_4$-tidy graph, we shall prove that it outputs that the graph is neighborhood-perfect if and only if it is $\{C_5, 3\text{-sun}, \overline{3K_2}\}$-free. This together with Theorem~\ref{teo:CaractOfP4TidyNP} will imply the correctness of the algorithm. 

    Suppose that Algorithm~\ref{algo:RecP4TidyNP} outputs that $G$ is not neighborhood-perfect. Hence the algorithm stopped in Step 1. If it stopped in line~\ref{ln:2C5}, then clearly $G$ contains an induced $C_5$ or $3$-sun. stopped in line~\ref{ln:2Starfish} or line~\ref{ln:2SnodeThreeChildren}, then $h$ is an S-node and consequently each of its nonleaf children represents a non-complete graph. On the one hand, if the algorithm stopped in line~\ref{ln:2Starfish}, then any set consisting of a pair of nonadjacent vertices of each of the three nonleaf children of $h$ induces $2K_1 \vee 2K_1\vee 2K_1 =\overline{3K_2}$ in $G$. On the other hand, if the algorithm stopped in line~\ref{ln:2SnodeThreeChildren}, then the vertices of an induced $C_4$ of the graph represented by $h_1$ or $h_2$ together with a pair of nonadjacent vertices of the graph represented by the other one induce $C_4\vee 2K_1=\overline{3K_2}$ in $G$ as well. We conclude that if the algorithm outputs that $G$ is not neighborhood-perfect, then $G$ contains an induced $C_5$, $3$-sun or $\overline{3K_2}$.

    Let us now prove that, conversely, if $G$ contains any of the three forbidden induced subgraphs, then the algorithm outputs that $G$ is not neighborhood-perfect. 

    Suppose first that $G$ contains an induced $C_5$ or an induced $3$-sun. By Theorem~\ref{theo:primeFreeGraphs}, there is some N-node $h$ of $T(G)$ such that $\pi(h)$ contains an induced $C_5$ or $3$-sun. By Theorem~\ref{theo:StructureOfP4Tidy}, $\pi(h)$ is $C_5$ or an urchin with at least three ends and the algorithms outputs that $G$ is not neighborhood-perfect in line~\ref{ln:2C5}. Finally, let us consider the case when $G$ contains an induced $\overline{3K_2}$. Let $h$ be a node of $T(G)$ such that $G[h]$ contains an induced $\overline{3K_2}$ but none of the graphs represented by its children does. Clearly $h$ cannot be a P-node, so it must be an N-node or S-node. If $h$ is an N-node and $G[h]$ contains an induced $\overline{3K_2}$, then $\pi(h)$ must be an urchin or a starfish. However, if $\pi(h)$ were an urchin, it would contain a $3$-sun. Hence, without loss of generality, let us suppose that it is not an urchin. Suppose $\pi(h)$ is a starfish with partition $(S,C,R)$ with the nodes of $S$ and $C$ being leafs of $T(G)$ and $R$ consisting on a single node $h_r$. By hypothesis, clearly the $\overline{3K_2}$ cannot be entirely in $h_r$, $C$, or $S$. Since every vertex of a graph represented by a node in $S$ has degree at most $2$ in $G[h]$, no vertices of graphs represented by nodes in $S$ can be vertices of any induced $\overline{3K_2}$ of $G[h]$. Now, as each vertex of a graph represented by a vertex of $C$ are adjacent to every vertex of the graph represented by $h_r$, and $\overline{3K_2}$ has no universal vertex, then each induced $\overline{3K_2}$ must have at least two nonadjacent vertices belonging to graphs represented by a vertex of $C$. But this is only possible if $G[h]$ is a fat urchin where some node of $C$ represents $2K_1$. If this is the case, then an induced $\overline{3K_2}$ can only be formed if there is an induced $C_4$ in the graph represented by $h_r$. To conclude if $h$ is an S-node, since $\overline{3K_2} = C_4 \vee 2K_1 = 2K_1 \vee 2K_1 \vee 2K_1$, the only two possibilities for $G[h]$ to have an induced $\overline{3K_2}$ while none of its children have it, are that there are more than three children representing non-complete graphs or two children, one containing a $C_4$ and the other one representing a non-complete graph. In all cases the algorithm outputs that $G$ is not neighborhood-perfect, which completes the proof of the correctness of the algorithm.

    The time complexity of the algorithm can easily be seen to be $\BigO(n+m)$ and $\BigO(n)$ if the decomposition tree is given. It was already mentioned in Section~\ref{sec:preliminaries} that the modular decomposition tree of $P_4$-tidy graphs can be found in linear-time and within the same time bound a partition of the N-nodes that correspond to urchins or starfish can be given. It was already proven in Theorem~\ref{teo:C4TestCorrectness} that Algorithm~\ref{algo:C4Test} runs in $\BigO(n)$ time. In Step $1$ we traverse every node $h$ of $T(G)$, and all the operations corresponding to each node $h$ can be carried out in $\BigO(n_\pi(h))$ time once that $C(h)$ has been determined. Hence, as the sum of $n_{\pi}(h)$ over all nodes $h$ is at most $2n$, the algorithm runs in $\BigO(n+m)$ time and even in $\BigO(n)$ time if the modular decomposition tree $T(G)$ is already given in the input.
 \end{proof} 

Having presented the algorithm for $P_4$-tidy graphs we shall give another one to decide neighborhood-perfectness of tree-cographs in linear-time.

As was already pointed out in Section~\ref{sec:preliminaries}, the N-nodes of the modular decompositions of tree-cographs represent only trees and complement of trees. Moreover neighborhood-perfect tree-cographs were characterized in Theorem~\ref{teo:CaractNPTreeCograph} as tree-cographs having no $\overline{3K_2}$ or $P_6 \vee 3K_1$ as induced subgraphs. We shall use this characterization and the modular decomposition of tree-cographs to achieve a linear-time recognition algorithm.

We shall first define two functions defined on the nodes $h$ of the modular decomposition tree $T(G)$ of a graph $G$. Let $P : V(T(G)) \to \{True, False\}$, such that $P(h) = True$ if and only if $G[h]$ has an induced $P_6$. And let $\alpha : V(T(G)) \to \mathbb{N}$, such that $\alpha(h) = \alpha(G[h])$.

Algorithm~\ref{algo:P6_C4_alpha_geq_3_test} computes both $P(h)$ and $\alpha(h)$ for all nodes in a modular decomposition tree $T(G)$ of a tree-cograph. It computes as well $C(h)$ as was defined in above, all in $\BigO(n)$ time, given the modular decomposition tree. It uses the fact that computing $\alpha(T)$ can be done in time $\BigO(|V(T)|)$, for any tree $T$ \cite{DFS_for_vertex_cover_Savage_1982}.  

\begin{algorithm2e}[ht!]
\small
\DontPrintSemicolon
\SetKwInput{Input}{Input}
\SetKwInput{Output}{Output}
\SetKwBlock{StepOne}{Step 1:\; Traverse the nodes of $T(G)$ in post-order, and in each node $h$ \textbf{do}:}{end}
\SetKwBlock{StepTwo}{Step 2:}{end}

\Input{A $P_4$-tidy graph $G$ and its modular decomposition tree $T(G)$}
\Output{$C(h)$, $P(h)$ and $\alpha(h)$ for every node $h$ of $T(G)$}    

\StepOne{
    
    \lIf{$h$ is a leaf}{ $C(h) := P(h) := False$ and $\alpha(H) := 1$} 

    \ElseIf{$h$ is a P-node with children $h_1, \dots, h_k$}{
        $C(h) := \bigvee_{i=1}^k C(h_i)$,   $P(h) := \bigvee_{i=1}^k P(h_i)$,   $\alpha(h) := \sum_{i=1}^k \alpha(h_i)$
    }  

    \ElseIf{$h$ is a S-node with children $h_1, \dots, h_k$}{
        $\alpha(h)  := \max\{ \alpha(h_i) \colon 1\leq i \leq k\}$, $P(h) := \bigvee_{i=1}^k P(h_i)$,\;
        \lIf{ $h$ has at least two nonleaf children}{$C(h) := True$
        }\lElse{$C(h) := \bigvee_{i=1}^k C(h_i)$}
    }

    \ElseIf{$\pi(h)$ is a tree with children $h_1, \dots, h_k$}{
        compute $\alpha(G[h])$ in linear-time and assign it to $\alpha(h)$, $C(h) :=  False$ \;
    \lIf{the longest path in $\pi(h)$ is of length at least $6$}{$P(h) := True$}
    \lElse{$P(h) := False$}

    }

    \ElseIf{$\pi(h)$ is a co-tree with children $h_1, \dots, h_k$}{
        $\alpha(h) := 2$, $P(h) := False$\;
        \If{$\overline{\pi(h)}$ has an induced matching of size at least $2$}{ $C(h) := True$}
        \lElse{$C(h) := False$}
    }

}

\StepTwo{Output $C(h)$, $P(h)$ and $\alpha(h)$ for every node $h$ of $T(G)$}

\caption{Computes $\alpha(h)$, $P(h)$ and $C(h)$ for every node $h$ of $T(G)$, with $G$ a tree-cograph  \label{algo:P6_C4_alpha_geq_3_test}}
\end{algorithm2e}
Algorithm~\ref{algo:RecTreeCographNP} is a linear-time algorithm, that uses Algorithm~\ref{algo:P6_C4_alpha_geq_3_test} to determine whether any given tree-cograph is neighborhood-perfect.

\begin{algorithm2e}[ht!]
\small
\DontPrintSemicolon
\SetKwInput{Input}{Input}
\SetKwInput{Initialization}{Initialization}
\SetKwInput{Output}{Output}
\SetKwBlock{StepOne}{Step 1:\; Traverse every node $h$ of $T(G)$ in any order and \textbf{do}:}{end}
\SetKwBlock{StepTwo}{Step 2:}{end}

    \Input{A tree-cograph $G$}
    \Output{Determines whether $G$ is a neighborhood-perfect graph}
    \Initialization{Build the modular decomposition tree $T(G)$ of $G$ and compute $C(h)$, $P(h)$, and $\alpha(h)$ for every node $h$ of $T(G)$ using Algorithm~\ref{algo:P6_C4_alpha_geq_3_test}}
    
\StepOne{
        \If{$h$ is an S-node}{
            
            \If{$h$ has at least three nonleaf children}{
                \textbf{output} ``$G$ is not neighborhood-perfect'' and \textbf{stop} \label{ln:4Snode3Children}
            }
            \ElseIf{$h$ has exactly two nonleaf children $h_1$ and $h_2$}{
                \If{$C(h_1)$ or $C(h_2)$ is True}{
                    \textbf{output} ``$G$ is not neighborhood-perfect'' and \textbf{stop} \label{ln:4Snode2Children}
                }
                \ElseIf{$P(h_1)$ is True and $\alpha(h_2)\geq 3$ or vice versa}{
                    \textbf{output} ``$G$ is not neighborhood-perfect'' and \textbf{stop} 
                    \label{ln:4SnodeWithP6vee3K1}
                }
            }

        }

        \If{$h$ is an N-node, with $\pi(h)$ a co-tree}{
            \If{$\overline{G[h]}$ contains an induced matching of size at least 3}{
            \textbf{output} ``$G$ is not neighborhood-perfect'' and \textbf{stop} \label{ln:4CotreeWith3Pyr}
            }
        }
    
}
\StepTwo{ \textbf{output} ``$G$ is neighborhood-perfect''}
\caption{Recognition of neighborhood-perfectness of tree-cograph\label{algo:RecTreeCographNP}}
\end{algorithm2e}

We shall proceed to prove that both Algorithm~\ref{algo:P6_C4_alpha_geq_3_test} and Algorithm~\ref{algo:RecTreeCographNP} are both correct and run in the previously stated time bounds. 

\begin{theorem} \label{teo:CorP6C4AlphaTest}
    Algorithm~\ref{algo:P6_C4_alpha_geq_3_test} correctly computes $C(h)$, $P(h)$ and $\alpha(h)$ for every node $h$ of a given modular decomposition tree $T(G)$ in $\BigO(n + m)$ time, whenever $G$ is a tree-cograph.
\end{theorem}
\begin{proof}
     The nodes of $T(G)$ are traversed in post-order, meaning that when the algorithm computes the functions $C$, $P$, and $\alpha$ for $h$, all the children of $h$ have already been processed. It is clear that if $h$ is a leaf, the functions are correctly computed. Let us prove then that for each node $h$ that is not a leaf, the functions are correctly computed, assuming they were correctly computed for the children of $h$. 

     If $h$ is a P-node, then clearly the maximum independent set of $G[h]$ is the union of the maximum independent sets of each component, moreover it contains an induced $P_6$ or $C_4$ if and only if one of the components has one. 

     If $h$ is an S-node, then clearly the maximum independent set of  $G[h]$ is an independent set of one of the graphs represented by its children. It is as well clear that as $P_6$ has a connected complement, it must be contained in one of the components of $\overline{G[h]}$, which are the graphs represented by the children of $h$. As for the $C_4$, since it can be formed by the join of two $2K_1$, it can be and induced subgraph of $G[h]$ if and only if it is and induced subgraph of the graph represented by some of the children of $h$ or if there are two nonleaf children of $G$ (because the join of one non-edge from each of the graphs represented by them form an induced $C_4$). Thus the only case that remains to be considered is when $h$ is an N-node. 

     If $h$ is an N-node, with $\pi(h)$ a tree, then, as $G[h]$ is a tree, it cannot contain an induced $C_4$, and it contains an induced $P_6$ if and only if there are two vertices at distance $5$ or more. If $h$ is an N-node and $\pi(h)$ is a co-tree with connected complement, then $\alpha(G[h]) = 2$ because it cannot be greater than $2$ ($\overline{\pi(h)}$ would contain a $C_3$) and if it were $1$, then $\pi(h)$ would be complete and therefore have a disconnected complement. Similarly $G[h]$ cannot contain an induced $P_6$, because it has three independent vertices that would form a $C_3$ in the complement of $\pi(h)$. Finally as $C_4 = \overline{2K_2}$, $\pi(h) = G[h]$ contains an induced $C_4$ if and only if $\overline{\pi(h)}$ contains an induced matching of size $2$.

     To prove that the algorithm runs in $\BigO(n + m)$ time, we shall see that for every node of $T(G)$, it performs $\BigO(n_{\pi}(h))$ operations, except for the N-nodes $h$ with $\pi(h)$ isomorphic to a co-tree, in which the number of operations is in $\BigO(n_{\pi}(h) + m_{\pi}(h))$. As mentioned in Section~\ref{sec:preliminaries} the sum of $n_{\pi}(h)$ over all nodes of $T(G)$ is at most $2n$, since all edges in $\pi(h)$, for $h$ an N-node are in one-to-one correspondence with edges of $G$, and two graphs represented by two different N-nodes are vertex-disjoint, the sum of $m_{\pi}(h)$ for all N-nodes with $\pi(h)$ a co-tree must be at most $m$. 

     It is clear that if $h$ is a leaf, a P-node, or an S-node, then the number of operations is proportional to $n_{\pi}(h)$. If $h$ is an N-node with $\pi(h)$ isomorphic to a tree, then using any of the algorithms in \cite{LinearAlgTreesPreeceding_Mitchell1975,LinearAlgTreesRecursiveMithcell1979,DFS_for_vertex_cover_Savage_1982} a maximum cardinality independent set can be found in $\BigO(n_{\pi}(h))$ time. And using the algorithm, suggested by Dijkstra in the sixties and formally proved in \cite{LongestPathInTreeBulterman2002}, to find a maximum path in trees it can be easily tested if the longest path in $\pi(h)$ has size greater or equal to $6$ in $\BigO(n_{\pi}(h))$ time. The last case to consider is the if $h$ is an N-node, with $\pi(h)$ isomorphic to a co-tree. Because $\pi(h)$ has $\Theta((n_{\pi}(h)^2)$ edges, then in $\BigO(m_{\pi}(h))$ time it can be complemented. Once complemented, in $\BigO(n_{\pi}(h)$ time the size of the greatest induced matching can be determined using any of the algorithms in \cite{FrickeStronMatchingOnTrees1992,GolumbicInducedMatching2000,ZitoMaxInducedMatchingTrees2000}. This fact together with the observations made in the last paragraph imply that the whole algorithm can be implemented to run in $\BigO(n + m)$ time.
\end{proof}

\begin{theorem}\label{teo:CorRecTreeCographNP}
    Algorithm~\ref{algo:RecTreeCographNP} correctly determines whether any given tree-cograph $G$ is neighborhood-perfect, in $\BigO(n + m)$ time.
\end{theorem}
\begin{proof}
     To prove the correctness of this algorithm, we shall apply the same reasoning as in the proof of Theorem~\ref{teo:CorRecP4TidyNP}, but using the subgraph characterization of neighborhood-perfect graphs among tree-cographs proved in Theorem~\ref{teo:CaractNPTreeCograph}. We will then prove that the algorithm outputs that the graph $G$ is neighborhood-perfect if and only if $G$ is $\{P_6 \vee 3K_1, \overline{3K_2}\}$-free.

     Let see first that if the algorithm outputs that the graph is not neighborhood-perfect, then it must contain one of the forbidden induced subgraphs. It must stop in Step 1. If it stops in line~\ref{ln:4Snode3Children}, then clearly $G[h]$ must contain an induced $\overline{3K_2} = 2K_1 \vee 2K_1 \vee 2K_1$; if it stops in line~\ref{ln:4Snode2Children} then it must contain an induced $C_4 \vee 2K_1 = \overline{3K_2}$. Moreover if it stops in line~\ref{ln:4SnodeWithP6vee3K1}, then one of the two children of $h$ contains an induced $P_6$ and the other one has an independent set of size at least $3$, implying that $G[h]$ contains an induced $P_6 \vee 3K_1$. Lastly if it stops in line~\ref{ln:4CotreeWith3Pyr}, then $G[h]$ is a co-tree that contains an induced $\overline{3K_2}$.

     To conclude the if and only if proof, suppose now that $G$ contains one of the two forbidden induced subgraphs, and let us see that the algorithm must then output that $G$ is not neighborhood-perfect. Clearly if $G$ contains one of the forbidden induced subgraphs, then there must be a node $h$ of $T(G)$ such that $G[h]$ contains the induced subgraph, but none of its children does. Clearly $h$ cannot be a P-node. Moreover, $h$ cannot be an N-node with $\pi(h)$ isomorphic to a tree, because both forbidden graphs have cycles. Thus $h$ must be a S-node or an N-node with $\pi(h)$ isomorphic to a co-tree. If $h$ is a S-node and $G[h]$ contains an induced $\overline{3K_2}$, then, as was shown in the proof of Theorem~\ref{teo:CorRecP4TidyNP}, either $h$ has three nonleaf children or has exactly two nonleaf children one of which contains an induced $C_4$. On the other hand if $h$ is an S-node but $G[h]$ contains an induced $P_6 \vee 3K_1$, then as both $P_6$ and $3K_1$ are not the join of any other graph, there must be two children of $h$, one representing a graph containing an induced $P_6$ and the other one a graph having an independent set of size at least $3$. All of these cases are considered in lines~\ref{ln:4Snode2Children}, \ref{ln:4Snode3Children}, and \ref{ln:4SnodeWithP6vee3K1}. Finally if $h$ is an N-node, with $\pi(h)$ a co-tree, then clearly $G[h]$ cannot contain an induced $3K_1 \vee P_6$, because the complement of a $3K_1$ would be a $C_3$, and $G[h]$ is a co-tree. If i$G[h]$ contains an induced $\overline{3K_2}$, then it could only be because in the complement of $\pi(h)$ there is an induced $3K_2$, which is the same as saying that $\overline{\pi(h)}$ has an induced matching of size at least $3$. Again this is tested in line~\ref{ln:4CotreeWith3Pyr}. So we have proved that if $G$ contains one of the forbidden induced subgraphs, then the algorithm outputs that $G$ is not neighborhood-perfect, concluding the proof of the if and only if. 

     To see that the algorithm runs in $\BigO(n + m)$ time, we shall use the same argument as in Theorem~\ref{teo:CorP6C4AlphaTest}. First recall that as was mentioned in Section~\ref{sec:preliminaries} we can construct the modular tree in linear-time and, as was already proven, run Algorithm~\ref{algo:P6_C4_alpha_geq_3_test} in linear-time. Now, for every node $h$ in $T(G)$, if $h$ is a P-node or an N-node with $\pi(h)$ a tree, the algorithm does no operations. If $h$ is an S-node, then it clearly can determine the number of children $h_i$ of $h$ and check the values of $C(h_i)$, $P(h_i)$ and $\alpha(h_i)$ for all of them, in $\BigO(n_{\pi}(h))$ time. Finally if $h$ is an N-node, with $\pi(h)$ a co-tree, then, as $m_{\pi}(h) \in \Theta(n_{\pi}(h)^2)$, we can complement $\pi(h)$ in $\BigO(m_{\pi}(h))$ time. Once complemented, we can use any of the linear-time algorithms in \cite{FrickeStronMatchingOnTrees1992,ZitoMaxInducedMatchingTrees2000,GolumbicInducedMatching2000}, to compute a maximum induced matching of $\overline{\pi(h)}$ in $\BigO(n_{\pi}(h))$ time. Thus the algorithm makes at most a number of operations proportional to $n_{\pi}(h)$ for every N-node $h$ and to $m_{\pi}(h)$ for the N-nodes with $\pi(h)$ a co-tree, which implies that it runs in $\BigO(n + m)$ time for the whole graph.
\end{proof}

\subsection{Optimal Sets Algorithms}\label{subsec:optimal}
Clearly, a maximum neighborhood-independent set and minimum neighborhood covering set of the disjoint union of two graphs can be obtained by the union of the respective sets in the smaller graphs. This, together with Theorems~\ref{teo:AnOfJoin} and \ref{teo:PnOfJoin}, and the particular modular decomposition trees of $P_4$-tidy graphs and tree-cographs, allowed us to prove the following results.

In this section we shall present two new linear-time algorithms to compute a maximum neighborhood-independent set, a minimum neighborhood set, a maximum $2$-independent set, and a minimum dominating set of $P_4$-tidy and tree-cographs. We will refer a maximum neighborhood-independent set, a minimum neighborhood set, a maximum $2$-independent set and a minimum dominating set of a graph as \emph{optimal sets} of the graph. As in the previous section, we shall strongly use the properties of the modular decomposition trees of these two classes. 

First we shall present an algorithm that given a subroutine that computes the optimal sets of graphs represented by the N-nodes of the modular decomposition tree (meaning that it computes a maximum neighborhood-independent set, a minimum neighborhood independent set, a domination sets), finds optimal sets for the graphs represented by all the remaining nodes of the modular decomposition tree. This algorithm will be used for both classes of graphs, changing only the routine that finds optimal sets for the graph represented by the N-nodes (which have a different characterization in each class). It is also interesting to note that given any other graph class with a known characterization of its modular decomposition tree, one needs only to find a routine that finds optimal sets for the graph represented by the N-nodes from optimal sets of its children, to obtain an algorithm that finds optimal sets in the whole graph.

Given a graph $G$ and its modular decomposition tree $T(G)$, for any node $h$ of $T(G)$, let $R_{\mathrm n}(h)$ be a list of vertices of $G$ that form a neighborhood set of $G[h]$ of minimum size, $A_{\mathrm n}(h)$ be a list of vertices and edges forming a maximum neighborhood independent set of $G[h]$, $A_2(h)$ be a list of vertices forming a $2$-independent set of maximum size of $G[h]$, and $D(h)$ be a list of vertices of $G$ constituting a minimum dominating set of $G[h]$. We will call these four lists, \emph{optimal lists} for the node $h$. Algorithm~\ref{algo:GeneralOptimalSet} will show how to recursively obtain optimal lists for each node $h$, thus obtaining these lists for the root of $T(G)$, which we shall call $A_{\mathrm n}(G)$, $R_{\mathrm n}(G)$, $A_2(G)$ and $D(G)$, respectively. For this purpose, Algorithm~\ref{algo:GeneralOptimalSet} will assume that we have a subroutine that given any N-node and optimal lists for the children of the N-node, correctly obtains the lists for the N-node. In all the following algorithms, we shall denote the concatenation of lists $l_1,\ldots,l_k$, with $1\leq i \leq k$ as $\sum_{i = 1}^k l_i$. To denote the concatenation of two lists $l_1$ and $l_2$, we will use $l_1+l_2$. We will denote a list by listing its elements between `$\langle$' and `$\rangle$'; for instance, a list whose elements are $x,y,z$ will be denoted by $\langle x,y,z\rangle$. If $l$ is a list, we will denote by $l[i]$ its $i$-th element.

\begin{algorithm2e}[htbp!]
\small
\DontPrintSemicolon
\SetKwInput{Input}{Input}
\SetKwInput{Output}{Output}
\SetKwInput{Initialization}{Initialization}
\SetKwBlock{StepOne}{Step 1: \; Traverse the nodes of $T(G)$ in post-order, and in each node $h$ \textbf{do}:}{end}
\SetKwBlock{StepTwo}{Step 2:}{end}
\Input{A graph $G$}
\Output{$A_{\mathrm n}(G)$, $R_{\mathrm n}(G)$, $A_2(G)$ and $D(G)$}    
\Initialization{Construct $T(G)$, the modular decomposition tree of $G$}

\StepOne{

    \If{$h$ is a leaf, representing only $v\in V(G)$}{ $A_{\mathrm n}(h) := \langle v\rangle$, $R_{\mathrm n}(h):= \langle v\rangle$, $A_2(h):= \langle v\rangle$, $D(h):= \langle v\rangle$} 

    \ElseIf{$h$ is a P-node with children $h_1, \dots, h_k$}{ $R_{\mathrm n}(h) := \sum_{i=1}^k R_{\mathrm n}(h_i)$, $A_2(h) := \sum_{i=1}^k A_2(h_i)$, $D(h) := \sum_{i=1}^k D(h_i)$, $A_{\mathrm n}(h) := \sum_{i=1}^k A_{\mathrm n}(h_i)$}  

    \ElseIf{$h$ is an S-node with children $h_1, \dots, h_k$}{
        
        $A_2(h):=\langle v\rangle$ where $v$ is an arbitrary vertex of $G[h]$\;
   
        $D(h):=$ a list of minimum length among $D(h_1),\ldots,D(h_k),\langle v_1,v_2\rangle$ for any $v_1\in V(G[h_1])$ and $v_2\in V(G[h_2])$\;
        
        \If { $k=2$ }{$ A_{\mathrm n}(h):= \langle (A_2(h_1)[i],A_2(h_2)[i])\colon {1\leq i\leq\min\{\vert A_2(h_1)\vert,\vert A_2(h_2)\vert\}\rangle}$ \label{ln:Alg5EdgeFormation}}
        \Else{
               $A_{\mathrm n}(h):=\langle (v_1,v_2) \rangle$ for any $v_1\in V(G[h_1])$ and $v_2\in V(G[h_2])$\;
            }

        \begin{tabular}{rl}   
        $R^*:={}$ & a list of minimum length among\\
         &$D^*(h_1),\ldots,D^*(h_k),R_{\mathrm n}(h_1),\ldots,R_{\mathrm n}(h_k)$,\\ 
         &where  $D^*(h_i)=D(h_i)+\langle v\rangle$ for any $v\in G[h]\setminus G[h_i]$\\
        \end{tabular}

        \If{ $k=2$ } {
            $R_{\mathrm n}(h):=R^*$}
        \Else{
            $R_{\mathrm n}(h):=$ a list of minimum length between $R^*$ and $\{v_1,v_2,v_3\}$, where $v_i\in V(G[h_i)])$ for $i\in\{1,2,3\}$ }
                }
    \ElseIf{$h$ is an N-node}{
            Use a graph class specific subroutine to calculate $A_{\mathrm n}(h)$, $R_{\mathrm n}(h)$, $A_2(h)$ and $D(h)$
    }

}

\StepTwo{Output $A_{\mathrm n}(G)$, $R_{\mathrm n}(G)$, $A_2(G)$, $D(G)$}

\caption{Computes $A_{\mathrm n}(G)$, $R_{\mathrm n}(G)$, $A_2(G)$, $D(G)$ of a graph $G$, if a subroutine to find optimal lists for the graphs represented by N-nodes of its modular decomposition tree is given  \label{algo:GeneralOptimalSet}}

\end{algorithm2e}

Below, we prove that Algorithm~\ref{algo:GeneralOptimalSet} correctly calculates the desired lists, given that the subroutine used to calculate the lists in the N-nodes works correctly. Moreover we shall prove that if the N-nodes' subroutine works in linear time with respect to $\pi(h)$ for an N-node $h$, then the Algorithm~\ref{algo:GeneralOptimalSet} works in linear time with respect to $G$.

\begin{theorem}
    Algorithm~\ref{algo:GeneralOptimalSet} obtains correctly $A_{\mathrm n}(G)$, $R_{\mathrm n}(G)$, $A_2(G)$ and $D(G)$, given that the subroutine for N-nodes is correct.
\end{theorem}
\begin{proof}
     The algorithm traverses $T(G)$ in post-order, meaning that before reaching a node $h$, all its children have their optimal lists already computed. It is clear that if $h$ is a leaf, then $G[h]$ is a single vertex and then all optimal lists associated with $h$ consist in precisely that vertex. Let us now see that if we suppose the algorithm correctly builds optimal lists for all the children $h_1, \ldots, h_k$ of an S-node or P-node, then it correctly computes them for the node itself. If $h$ is a P-node, then, since the graphs represented by its children are the components of $G[h]$, it is clear that all optimal sets required can be obtained by simply joining the lists of the optimal sets for the children. If $h$ is an S-node, it is easy to see that each of the lists $A_2(h)$, $D(h)$, $A_{\mathrm n}(h)$, $R_{\mathrm n}(h)$ represent a dominating set, a $2$-independent set, a neighborhood-independent set and a neighborhood set of $G[h]$, respectively. Moreover, since the lengths of these lists match the optimal values (according to Theorems~\ref{teo:PnOfJoin}, \ref{teo:AnOfJoin}, \ref{lem:DomSetofMultJoin}, \ref{cor:PnOfMultJoin}, and \ref{cor:AnOfMultJoin}), the lists build in this way are optimal lists.
\end{proof}

\begin{lemma} \label{lem:boundOfEdgesMadeinAN}
    Let $c(h)$ be the number of edges made in line~\ref{ln:Alg5EdgeFormation} if $k=2$, for $h$ and all the descendants of $h$ in a modular decomposition tree $T(G)$. Hence, for every node $h$, $c(h) + \at{h} \leq n(h)$, where $\at{h} = \at{G[h]}$.
\end{lemma}
\begin{proof}
   To prove this statement, we shall use a structural induction in $T(G)$. First, let us see that for each leaf $h$, clearly $c(h) = 0$, $\at{h} = 1$, and $n(h) = 1$. Now, let us suppose that we have a node $h$, not a leaf, and that the statement holds for every child $h_i, 1\leq i \leq k$ of $h$. If $h$ is not an S-node, then clearly $c(h) = \sum_{i=1}^k c(h_i)$. As every $G[h_i] \subseteq G[h]$, the inequality $|I\cap V(h_i)| \leq \at{h_i}$ must hold for every $2$-independent set $I$ of $G[h]$. Hence, by the induction hypothesis, $c(h) + \at{h} \leq \sum_{i=1}^k c(h_i) + \sum_{i=1}^k \at{h_i} \leq \sum_{i=1}^k n(h_i) = n(h)$. 

   If $h$ is an S-node and $k > 2$, then $\at{h} = 1$, implying $c(h) + \at{h} = (\sum_{i=1}^k c(h_i)) + 1 \leq \sum_{i=1}^k c(h_i) + \at{h} \leq \sum_{i=1}^k n(h_i) = n(h)$. Hence, suppose that $h$ is an S-node with two children and suppose, without loss of generality, that $\at{h_1} \leq \at{h_2}$ and consequently $c(h) = c(h_1) + c(h_2) + \at{h_1}$. Thus, since $\at{h} = 1$ (because $h$ is an S-node), then $c(h) + \at{h} = c(h_1) + c(h_2) + \at{h_1} + 1 \leq c(h_1) + c(h_2) + \at{h_1} + \at{h_2} \leq n(h_1) + n(h_2) = n(h)$.  
\end{proof}

\begin{theorem}\label{teo:LinearGeneralOptimalSetAlgorithm}
    Algorithm~\ref{algo:GeneralOptimalSet} works in $\BigO$$(n + m)$ time, if the subroutine for N-nodes works in $\BigO$$(n_{\pi}(h) + m_{\pi}(h))$ time, for every N-node $h$.
\end{theorem}
\begin{proof}
    All nodes are traversed exactly once, so let us see that for every leaf, S-node and P-node, the algorithm performs $\BigO$$(n_{\pi}(h))$ operations. If $h$ is a leaf, then it only creates four lists of size $1$. If $h$ is a P-node, then the algorithm concatenates four times $n_{\pi}(h)$ lists. Which, if we suppose is done by loosing the original lists, can be achieved in $\BigO$$(n_{\pi}(h))$ time. If $h$ is an S-node, then to obtain $D(h)$, $A_2(h)$, $A_{\mathrm n}(h)$, and $R_{\mathrm n}(h)$, clearly it performs at most $\mathcal O(n_\pi(h))$ operations plus the time of building the edges in line~\ref{ln:Alg5EdgeFormation}, if $h$ is an S-node with exactly two children. Since, by Lemma~\ref{lem:boundOfEdgesMadeinAN}, the number of edges made in all S-nodes is $\BigO$$(n)$, the sum of $n_{\pi}(h)$ for every node $h$ in $T(G)$ is at most $2n$, the sum of all $m_{\pi}(h)$ for all N-nodes is at most $m$, and finding the modular decomposition tree can be done in time $\BigO$$(n + m)$, the whole algorithm can be implemented to run in $\BigO$$(n + m)$ time.
\end{proof}

Now that we have the ``general'' algorithm, we shall show an algorithm to find in $\BigO$$(n_\pi(h))$ time the optimal sets for an N-node of the modular decomposition tree $T(G)$ of a $P_4$-tidy graph $G$.

%%Meter rutina para N-nodes de P4 tidy

\begin{algorithm2e}[htbp!]
\small
\DontPrintSemicolon
\SetKwInput{Input}{Input}
\SetKwInput{Output}{Output}
\SetKwInput{Initialization}{Initialization}
\SetKwBlock{StepOne}{Step 1:}{end}
\SetKwBlock{StepTwo}{Step 2:}{end}

\Input{An N-node $h$ of a modular decomposition tree of a $P_4$-tidy graph $G$}
\Output{$A_{\mathrm n}(h)$, $R_{\mathrm n}(h)$, $A_2(h)$ and $D(h)$}    
\StepOne{
    
    \If{$\pi(h)$ is isomorphic to $C_5 = v_1\dots v_5v_1$}{ $A_{\mathrm n} := \langle v_1v_2, v_4v_5\rangle$, $R_{\mathrm n}(h) := \langle v_1, v_3, v_5\rangle$, $A_2(h) := \langle v_1\rangle$, $D(h) := \langle v_1,v_2\rangle$} 

    \ElseIf{$\pi(h)$ is isomorphic to $P_5 = v_1\dots v_5$}{ $A_{\mathrm n} := \langle v_1v_2, v_4v_5\rangle$, $R_{\mathrm n}(h) := D(h) := \langle v_2, v_4\rangle$, $A_2(h) := \langle v_1, v_4\rangle$}  

    \ElseIf{$\pi(h)$ is isomorphic to $\overline{P_5}$ with $\overline{\pi(h)} = v_1\dots v_5$}{ $A_{\mathrm n} := \langle v_1v_5, v_2v_3\rangle$, $R_{\mathrm n}(h) := D(h) = \langle v_1, v_2\rangle$, $A_2(h) = \langle v_1\rangle$}

    \ElseIf{$\pi(h)$ is a starfish with partition $(S,C,R)$ where $C=\{c_1,\ldots,c_k\}$, $S=\{s_1,\ldots,s_k\}$ and $c_1s_1,\ldots,c_ks_k$ are the legs of $\pi(h)$}{
        Let $v_i\in V(G[c_i])$ and $w_i\in V(G[s_i])$ for each $i\in\{1,\ldots,k\}$\;
        $A_2(h) := \langle w_1,\ldots,w_k\rangle$, $D(h) := \langle v_1,\ldots,v_k\rangle$, $A_{\mathrm n}(h):=\langle v_1w_1,\ldots,v_kw_k\rangle$, $R_{\mathrm n}(h):=\langle v_1,v_2,\ldots,v_k\rangle$\;
        \If{$\pi(h)$ is a fat starfish with $c_i\in C$ representing $2K_1$}{Replace $v_i$ in $R_{\mathrm n}(h)$ with $w_i$}
    }  

    \ElseIf{$\pi(h)$ is an urchin with partition $(S,C,R)$ where $C=\{c_1,\ldots,c_k\}$, $S=\{s_1,\ldots,s_k\}$ and $c_1s_1,\ldots,c_ks_k$ are the legs of $\overline{\pi(h)}$}{
        $A_2(h):=\langle v_1\rangle$, $D(h):=\langle v_1,v_2\rangle$, $A_{\mathrm n}:=\langle v_1w_2\rangle$ and $R_{\mathrm n}(h):=\langle v_1,v_2\rangle$ for any $v_1\in V(G[c_1])$, $v_2\in V(G[c_2])$ and $w_2\in V(G[s_2])$.
    }  

}

\StepTwo{Output $A_{\mathrm n}(h)$, $R_{\mathrm n}(h)$, $A_2(h)$, $D(h)$}

\caption{Computes $A_{\mathrm n}(h)$, $R_{\mathrm n}(h)$, $A_2(h)$, $D(h)$, for a given N-node $h$ of a modular decomposition tree $T(G)$ of a $P_4$-tidy graph $G$ \label{algo:NnodeP4TidyOptimalSet}}

\end{algorithm2e}

\begin{theorem}\label{teo:NnodeForP4TidyCorrect}
    Algorithm~\ref{algo:NnodeP4TidyOptimalSet} correctly finds $A_{\mathrm n}(h)$, $R_{\mathrm n}(h)$, $A_2(h)$, $D(h)$, for any given N-node $h$ of the modular decomposition tree $T(G)$ of any $P_4$-tidy graph $G$.
\end{theorem}
\begin{proof}
It can be checked by simple inspection that if $\pi(h)$ is isomorphic to $C_5$, $P_5$, or $\overline{P_5}$, optimal lists are chosen (recall that if this is the case $G[h] = \pi(h)$). If $\pi(h)$ is a starfish, then clearly the lists $A_2(h)$, $D(h)$, $A_{\mathrm n}(h)$ and $R_{\mathrm n}(h)$ computed by the algorithm correspond to a $2$-independent set, a dominating set, a neighborhood-independent set and a neighborhood set of $G[h]$, respectively. Moreover, such lists are optimal lists because $A_2(h)$ has the same length as $D(h)$, and $A_{\mathrm n}(h)$ has the same length as $R_{\mathrm n}(h)$. If $\pi(h)$ is an urchin, clearly we cannot dominate all vertices with only one vertex, but if we take two vertices belonging to different graphs represented by vertices of $C$, we obtain a minimum dominating set, as well as a minimum neighborhood set. In an urchin all vertices are at most at distance two from each other, and thus all $2$-independent sets of $G[h]$ have size $1$. It is also easy to see that if $\pi(h)$ is an urchin then no two edges can be neighborhood-independent; so, the maximum neighborhood-independent set must be composed of at most one edge. Hence, if $\pi(h)$ is an urchin, then the lists $A_2(h)$, $D(h)$, $A_{\mathrm n}(h)$ and $R_{\mathrm n}(h)$ build by the algorithm are optimal lists. Therefore, as $G$ is $P_4$-tidy, we have seen that for all possible scenarios the algorithm correctly computes the optimal sets.  
\end{proof}

\begin{theorem}\label{teo:NnodeForP4TidyComplex}
    Algorithm~\ref{algo:NnodeP4TidyOptimalSet} works in $\BigO(n_{\pi}(h))$ time, for $h$ an N-node in the modular decomposition tree of any $P_4$-tidy graph $G$.
\end{theorem}
\begin{proof}
    As we was already seen in Section\ref{sec:preliminaries}, if $G$ is a $P_4$-tidy, we can decide in $\BigO(n_{\pi}(h))$ time whether $\pi(h)$ is isomorphic to $P_5$, $C_5$, $\overline{P_5}$, or is a starfish or urchin, and in the latter two cases obtain its decomposition. It is clear that if $\pi(h)$ is isomorphic to a $P_5$, $C_5$, $\overline{P_5}$, the algorithm performs a constant number of operations. If $\pi(h)$ is a starfish, then once it has obtained $C$ and $S$, and determined if there is a replaced vertex of $C$ (all in $\BigO$$(n_{\pi}(h))$) time, it does only constant time assignments and it generates $|C|$ edges, all of which can be done in $\BigO$$(n_{\pi}(h))$ time. Finally if $\pi(h)$ is an urchin, then once again it performs a constant number of operations. Therefore in all possible cases it runs in $\BigO(n_{\pi}(h))$ time.
\end{proof}

Now we shall present an algorithm to find the optimal sets of N-nodes in a modular decomposition tree of a tree-cograph. To this purpose we shall first give a characterization of co-trees with $\alpha_{\mathrm n}>1$. This characterization will allow us to easily identify these graph and find a neighborhood independent set of maximum size, all in linear-time. 

\begin{lemma} \label{lem:TotDomSetInTree}
    If $G$ is a graph, then $\an{G} > 1$ if and only if $G$ has two edges $xy$ and $wz$ such that $\{x,y,w,z\}$ is a total dominating set of $\overline{G}$. 
\end{lemma}
\begin{proof}
   By definition, $\an{G} > 1$, if and only if there are two neighborhood-independent edges $xy$ and $wz$ in $G$. Moreover, any two edges $xy$ and $zw$ of $G$, neighborhood-independent satisfy that every vertex is at least nonadjacent in $G$ to at least one vertex in $\{x,y,w,z\}$ different from itself, or equivalently $\{x,y,w,z\}$ is a total dominating set of $\overline{G}$.
\end{proof}

\begin{lemma} \label{lem:CotreewithAn2isPath}
    If $G$ is a co-tree with $\an{G} > 1$, then $T'$ must be a path, where $T$ is $\overline{G}$ and $T'$ is the graph $T$, with all its leafs erased. 
\end{lemma}
\begin{proof}
    If $G$ is a co-tree, then clearly $T$ is a tree and hence $T'$ must be also a tree. Let us suppose by contradiction that $T'$ is not a path. As paths are trees with at most two leafs, then $T'$ has by our supposition three different leafs $x$, $y$, $z$ and, as $|V(T')| > 2$, these three vertices must form an independent set of $T'$ (and thus of $T$). The fact that these three vertices are in $T'$ implies that they were not leafs in $T$, but as they are leafs of $T'$, then they must have been adjacent to leafs in $T$. Given a tree, all vertices adjacent to leafs must be in all total dominating sets, because they are the only vertices that can dominate the leafs. Hence $x$, $y$ and $z$ are in all total dominating sets of $T$. Since $\an{G} > 1$, Lemma~\ref{lem:TotDomSetInTree} implies that there must be a vertex $w$ such that $\{x,y,w,z\}$ is a total dominating set of $T$ and without loss of generality $xy$, $wz$ are edges of $G$. But this clearly implies a contradiction, because $z$ is not strongly dominated in $T$ by $\{x,y,w,z\}$. The contradiction came from the supposition that $T'$ was not a path.
\end{proof}

Before presenting the characterization, we shall state an inequality that will be used in the proof of Theorem~\ref{teo:caractCotreeAnEq2}. 

\begin{theorem}[\cite{ChellaliTotDomBound2006}]\label{theo:IneqOfTotDomNumber}
    The following inequality holds for every tree $T$: \[  \gamma_{\mathrm t}(T) \geq (n(T) + 2 - l(T))/2. \] Where $n(T)$ is $|V(T)|$, $l(T)$ is the number of leafs of $T$ and $\gamma_{\mathrm t}(T)$ is the total dominating number of $T$. 
\end{theorem}

\begin{theorem}\label{teo:caractCotreeAnEq2}
    If $G$ is a co-tree, then $\an{G} > 1$ if and only if $T'$ is either $P_2$, $P_3$, $P_4$, or $P_5$ or $P_6$ with no central vertex of $T'$ adjacent to a leaf of $T$, where $T = \overline{G}$ and $T'$ is the graph $T$ with all its leafs erased.
\end{theorem}
\begin{proof}
 Let us first prove that if $G$ is a co-tree with $\an{G}>1$, then $T'$ is as described above. By Lemma~\ref{lem:CotreewithAn2isPath}, $T'$ must be a path. Clearly $T'$ cannot have only $1$ vertex, because $T$ would be a star and $\an{G}$ would be one. As we have already seen in Lemma~\ref{lem:TotDomSetInTree}, $\gamma_{\mathrm t}(T) \leq 4$ if $T = \overline{G}$. Thus, Theorem~\ref{theo:IneqOfTotDomNumber} implies that $6 \geq n(T) - l(T)$, but $n(T) - l(T) = n(T')$. Therefore $T'$ is $P_i$ with $2 \leq i \leq 6$. If $T' = P_5$ or $T' = P_6$, then suppose by contradiction that there is a leaf in $T$ adjacent to any central vertex of $T'$. As was already mentioned in the proof of Lemma~\ref{lem:CotreewithAn2isPath}, this means that there is a central vertex of $T'$ that must be in every total dominating set of $T$, this is also always true for both leafs of $T'$. But then there cannot be a total dominating set of $T$ of size $4$, because all three vertices are nonadjacent in $T$ and there is no vertex that is adjacent to all three at the same time. This leads to a contradiction because we have already proved that $\gamma_{\mathrm t} \leq 4$. Hence no central vertex of $T'$ can be adjacent to a leaf of $T$ if $T'$ is a $P_5$ or $P_6$.

 To prove the converse implication, if $T'$ is $P_2$, $P_3$ or $P_4$, simply take all vertices of $T'$ plus two, one, or zero leafs of $T$, respectively, adjacent to different leafs of $T'$, and we shall have a total dominating set of $T$ of size $4$. If this set is $\{x,y,w,z\}$, then clearly we can always take $xy$ and $wz$ to be non-edges of $T$ and thus edges of $G$, and by Lemma~\ref{lem:TotDomSetInTree}, $\an{G}>1$. If $T'$ is  $P_5$ or  $P_6$, we can take all vertices of $T'$, except for the central vertices of the path. As no central vertex is adjacent to leafs of $T$, then clearly these four vertices must be a total dominating set of $T$. Once again it is easy to check that we can find two non-edges of $T$ among these four vertices, and therefore $\an{G} > 1$.   
\end{proof}

\begin{corollary} \label{cor:algFindNISofCoTreeAnEq2}
    It is easy to decide in $\BigO$$(n + m)$ time whether a co-tree $G$ has $\an{G} > 1$ and if so find a neighborhood-independent set of $G$ size $2$.
\end{corollary}
\begin{proof}
    We use the characterization presented in Theorem~\ref{teo:caractCotreeAnEq2}. We can easily complement $G$ and remove the vertices with degree $1$. If the resulting tree is a path of length $2$ to $6$, then $\an{G}>1$ and, following the instructions of the proof of Theorem~\ref{teo:caractCotreeAnEq2}, we can obtain the two neighborhood-independent edges of $G$. As $G$ is a co-tree, $m \in \text{$\BigO$}(n^2)$ meaning that we can complement $G$ in time $\BigO$$(m)$. Deciding whether a tree becomes a path of bounded size by removing its leafs and, if so, also computing the corresponding path, can all be done in $\BigO$$(n)$. Finally obtaining the edges following the instructions of  Theorem~\ref{teo:caractCotreeAnEq2} can be easily done in time $\BigO$$(n)$.
\end{proof}

\begin{algorithm2e}[htbp!]
\DontPrintSemicolon
\SetKwInput{Input}{Input}
\SetKwInput{Output}{Output}
\SetKwInput{Initialization}{Initialization}
\SetKwBlock{StepOne}{Step 1:}{end}
\SetKwBlock{StepTwo}{Step 2:}{end}

\Input{An N-node $h$ of a modular decomposition tree of a tree-cograph $G$}
\Output{$A_{\mathrm n}(h)$, $R_{\mathrm n}(h)$, $A_2(h)$ and $D(h)$}    
\StepOne{
    \If{$\pi(h)$ is a tree}{
    $A_2(h):={}$a maximum $2$-independent set of $G[h]$\;
    $D(h):={}$a minimum dominating set of $G[h]$\;
    $A_{\mathrm n}(h):={}$a maximum matching of $G[h]$\;
    $R_{\mathrm n}(h):={}$a minimum vertex cover of $G[h]$\;
    }

    \ElseIf{$\pi(h)$ is a co-tree}{
        \lIf{$\overline{\pi(h)}$ has a total dominating set of size $2$ \label{ln:7totDomset}}{$A_2(h) :={}$ a total dominating set of $\overline{G[h]}$ of size $2$}
        \lElse{$A_2(h):= \langle v_1\rangle$ for any $v_1\in G[h]$}
        \lIf{$\an{\pi(h)} > 1$ \label{ln:7AnEq2}}{$A_{\mathrm n}(h) := \{e_1,e_2\}$ with $e_1, e_2$ neighborhood-independent edges of $G[h]$}
        \lElse{$A_{\mathrm n}(h) = \{e_1\}$ with $e_1$ any edge of $\pi(h)$}
        $D(h) := \langle v_l, v_n\rangle$, $R_{\mathcal n}(h) := \langle v_l, v_n\rangle$, with $v_l$ a leaf of $\overline{G[h]}$ and $v_n$ its only neighbor in $\overline{G[h]}$\;
    }
}

\StepTwo{Output $A_{\mathrm n}(h)$, $R_{\mathrm n}(h)$, $A_2(h)$, $D(h)$}

\caption{Computes $A_{\mathrm n}(h)$, $R_{\mathrm n}(h)$, $A_2(h)$, $D(h)$, for a given N-node $h$ of the modular decomposition tree $T(G)$ of a tree-cograph $G$ \label{algo:NnodeTreeCographOptimalSet}}

\end{algorithm2e}

Now that we have given this characterization, we shall prove that Algorithm~\ref{algo:NnodeTreeCographOptimalSet} finds the optimal sets for an N-node of the modular decomposition tree of a tree-cograph, all in $\BigO$$(n_{\pi}(h) + m_{\pi}(h))$ time. 

In line~\ref{ln:7totDomset}, we check if $\overline{\pi(h)}$ has a total dominating set of size $2$. Let us see why this allows us to find the $2$-independent set $\pi(h)$ that we need.

\begin{lemma}\label{lem:2indepsetGistotdomsetInCoG}
    If $G$ is a graph, then $\{v_1,v_2\}\subseteq V(G)$ is a $2$-independent set of $G$ if and only if it is a total dominating set of $\overline{G}$. 
\end{lemma}

\begin{proof}
     The set $S = \{v_1,v_2\}$ is a $2$-independent set of $G$ if and only if $N_G[v_1] \cap N_G[v_2] = \emptyset$. But this means that in $\overline{G}$ no vertex can be nonadjacent to both $v_1$ and $v_2$, which is to say that all vertices of $\overline{G}$ must be adjacent to $v_1$ or $v_2$. Therefore $S$ is a $2$-independent set of $G$ if and only if $S$ is a total dominating set of $\overline{G}$ of size $2$. 
\end{proof}

\begin{theorem}\label{teo:NnodeForTreeCographCorrect}
    Algorithm~\ref{algo:NnodeTreeCographOptimalSet} correctly finds $A_{\mathrm n}(h)$, $R_{\mathrm n}(h)$, $A_2(h)$ and $D(h)$, for any given N-node $h$ of the modular decomposition tree of a tree-cograph $G$.
\end{theorem}
\begin{proof}
If $G$ is a tree-cograph, then an N-node $h$ of its modular decomposition is a tree with connected complement or a connected co-tree. In both cases $\pi(h)$ is isomorphic to $G[h]$, thus we can find the optimal sets analyzing $\pi(h)$. If $\pi(h)$ is a tree, then as was already seen in \cite{lehel_neighborhood_1986} a maximum matching of $G[h]$ is also a maximum neighborhood independent edge set and a minimum vertex cover is a minimum neighborhood cover set. Hence if $\pi(h)$ is a tree, then clearly the algorithm computes the correct values for the optimal sets. On the other hand if $\pi(h)$ is a co-tree, then as was seen in Lemma~\ref{lem:2indepsetGistotdomsetInCoG}, if we find a total dominating set of size $2$ in $\overline{G[h]}$, we will have a $2$-independent set of size $2$ of $G[h]$. Clearly a co-tree cannot have an independent set of size three, thus $\an{\pi(h)} \leq 2$. Clearly if there are no $2$-independent sets of size $2$, then any node is a maximum $2$-independent set. It was already stated in Corollary~\ref{cor:algFindNISofCoTreeAnEq2} that there is a linear-time algorithm to determine if $\an{\pi(h)} > 1$ and if this is the case to find a neighborhood independent set of size $2$. Thus in line~\ref{ln:7AnEq2}, we correctly obtain $A_{\mathrm n}(h)$. Note that $\an{G[h]} \leq 2$, because if we take a leaf of $\overline{G[h]}$ and its only neighbor in $\overline{G[h]}$, we clearly have a neighborhood set as well as a dominating set of $G(h)$. Moreover if there were a dominating set or neighborhood set of size $1$, then that would mean an isolated vertex in $\overline{G[h]}$, which would contradict the fact that it is a tree.  
\end{proof}

\begin{theorem}\label{teo:NnodeForTreeCographComplex}
    Algorithm~\ref{algo:NnodeTreeCographOptimalSet} can be implemented to run in $\BigO$$(n_{\pi}(h) + m_{\pi}(h))$ time.
\end{theorem}
\begin{proof}
    It is clear that in linear time it can be determined if $\pi(h)$ is a tree. Moreover, if $\pi(h)$ is not a tree, then it must be a co-tree because all N-nodes of a tree-cograph are trees or co-trees. If $\pi(h)$ is a tree, algorithms for finding minimum vertex cover sets, minimum dominating sets and maximum matchings in linear-time can be found in \cite{LinearAlgTreesPreeceding_Mitchell1975,LinearAlgTreesRecursiveMithcell1979,DFS_for_vertex_cover_Savage_1982}. Obtaining a $2$-independent maximum set of a tree can also be done efficiently with an algorithm very similar to the one mentioned in \cite{LinearAlgTreesRecursiveMithcell1979} for independent sets. We explicitly state here, for the sake of completion, this linear-time algorithm for finding a $2$-independent maximum set of a tree $T$:  

    Given a tree $T$, we regard it as a directed tree with an arbitrary root vertex $r$ and traverse its vertices in post-order. For every vertex $i$, we determine $\text{Use}(i)$ and $\text{NUse}(i)$, where $\text{Use}(i)$ is a maximum $2$-independent set using vertex $i$ and $\text{NUse}(i)$ is defined analogously but without using $i$. Clearly, if $i$ is not a leaf, then $\text{Use}(i) = i \cup \bigcup(\text{ NUse($j$) }: j\text{ is a child of $i$ })$ and $\text{NUse($i$) } = \bigcup( \max\{ \text{ Use($j$) }, \text{ NUse($j$) }: j \text{ is child of $i$ })$\}, where $\max\{A,B\}$ denotes a set with maximum number of vertices among $A$ and $B$. If $i$ is a leaf, then clearly $\text{Use($i$)}=\{i\}$ and $\text{NUse($i$)}=\emptyset$. Hence, $\text{ Use($i$) }$ and $\text{ NUse($i$) }$ for all vertices $i$ can be determined in overall linear-time. Finally, $\max\{\text{Use($i$), NUse($i$)} \}$, which is a maximum $2$-independent set of $T$, can be found in linear-time.

    Hence, using the algorithms mentioned above, which clearly run in $\BigO$$(n_{\pi}(h))$ time, we can obtain corresponding to a node $h$ whenever $\pi(h)$ is a tree. If $\pi(h)$ is a co-tree, then, as was already mentioned, we can complement it in $\BigO$$(m_{\pi}(h))$ time, then using any of the algorithms mentioned in \cite{LaskarTotalDom1984,ChellaliTotDomBound2006,HenningTotalDomMonograph2013}, we can obtain a maximum total dominating set of $\overline{\pi(h)}$, and if it is of size $2$, we can obtain the corresponding total dominating set of $\overline{G[h]}$ and assign it to $A_2(h)$ (bearing in mind that $\pi(h)$ and $G[h]$ are isomorphic). Using the algorithm mentioned in Corollary~\ref{cor:algFindNISofCoTreeAnEq2}, we can find in time $\BigO$$(n_{\pi}(h) + m_{\pi}(h))$ a maximum neighborhood independent set of $G[h]$. Finally, having already complemented $\pi(h)$, finding a leaf of $\overline{G[h]}$ and its neighbor can be done easily in linear-time. Therefore if $\pi(h)$ is a co-tree, the algorithm can also be implemented to run in $\BigO$$(n_{\pi}(h) + m_{\pi}(h))$ time.
\end{proof}

\subsection{Complexity Results}

Theorems~\ref{teo:LinearGeneralOptimalSetAlgorithm} and \ref{teo:NnodeForTreeCographComplex} imply that the problem of finding $\an{G}$ and $\pn{G}$ can be solved in linear-time if $G$ is the complement of a tree. Nevertheless, as was already stated, the problems of determining these two parameters for general graphs have been proven to be \NP-hard \cite{chang_algorithmic_1993}. We prove here that even if $G$ belong to the class of complement of bipartite graphs, that includes the class of complements of trees, these problems are \NP-hard.

\begin{theorem}\label{teo:npHardAnanPNforCoBip}
    It is \NP-hard to determine $\an{G}$ and $\pn{G}$  for any graph $G$ complement of a bipartite graph.
\end{theorem}

We shall denote complement of bipartite graphs as co-bipartite graphs. The proofs of Theorems~\ref{teo:NPhardPNforCobip} and \ref{teo:npHardAnforCoBip} together constitute a proof of Theorem \ref{teo:npHardAnanPNforCoBip}.

If $X$ and $Y$ are disjoint sets and $F\subseteq X \times Y$, we shall denote by $(X,Y,F)$ the co-bipartite graph with vertex set $X\cup Y$ where $X$ and $Y$ are cliques and the edges between $X$ and $Y$ are those in $F$.

\begin{theorem}\label{teo:npHardAnforCoBip}
    It is \NP-hard to determine the neighborhood independence number in co-bipartite graphs.
\end{theorem}

\begin{proof}
We shall prove the \NP-hardness of the problem, by showing a polynomial reduction of the problem of determining the size of a maximum independent set of a graph $H$. For that purpose, given any graph $H$, we will define a co-bipartite graph $G$ such that $\an{G} = \alpha(G)$.
    
Given any graph $H=(V,E)$, let $G=(X,Y,F)$ where $X=\{v'\colon v\in V\}$, $Y=V\cup E$ and $F=\{v'e\colon v\in V, e\in E\text{ and }v\text{ is incident to }e\}\cup\{v'v \colon v\in V\}$; that is, we connect every vertex in $Y$ to its copy in $X$ and every edge in $Y$ to the copies of its endpoints in $X$. Let us first note that as there are no isolated vertices in $G$, then in order to determine the neighborhood-independence number we can restrict our attention to those neighborhood-independent sets consisting only of edges. Moreover, being $X$ and $Y$ cliques, there is some maximum neighborhood-independent set having all its edges in $F$.
   
Given an independent set $S\in V$ of $H$, let $I$ be the subset of $F$ defined by $I=\{v'v\colon v \in S\}$. It is easy to see that $I$ is a neighborhood-independent set because given two different edges $v'v$ and $w'w$ of $I$, there is no vertex adjacent to all four vertices. In fact, the only vertices in $X$ adjacent to $v$ and $w$ are $v'$ and $w'$ respectively and if there were an element of $Y$ adjacent to $v$, $v'$, $w$, and $w'$, then it would necessarily be an edge $e$ of $H$ joining $v$ to $w$, which contradicts the fact that $S$ is an independent set of $H$. This contradictions proves that $I$ is a neighborhood-independent set and hence $\an{G}\geq\alpha(H)$.

Conversely, let $I$ be a neighborhood-independent set of edges in $G$ such that $I\subseteq F$. We shall see that $S = \{v\in V \colon v'y \in I\}$ is an independent set of $H$. Suppose, for a contradiction, that there is an edge $e$ of $H$ joining two vertices $v$ and $w$ of $S$. By definition, there are $y_1,y_2\in Y$ such that $v'y_1,w'y_2\in F$ and, by construction, $e$ is adjacent in $G$ to all the four endpoints of $v'y_1$ and $w'y_2$, which contradicts the fact that $F$ is a neighborhood-independent set. This contradictions shows that $S$ is an independent set of $H$ and therefore $\alpha(H)\geq\an{G}$. This completes the proof of the polynomial reduction of the maximum independent set problem to the maximum neighborhood-independent set problem in co-bipartite graphs.
\end{proof}

To prove the \NP-hardness of determining the neighborhood number of co-bipartite graphs, we will use the following result from \cite{DinurAndSafra_HardnessOfAproxVCAnnofMath_2005}. 

\begin{theorem}[\cite{DinurAndSafra_HardnessOfAproxVCAnnofMath_2005}] \label{teo:inaproxOfVC}
    Given a graph $G$, it is \NP-hard to approximate the Minimum Vertex Cover to within any factor smaller than $10\sqrt{5}-21 = 1.3606\dots$.
\end{theorem}

\begin{theorem}\label{teo:NPhardPNforCobip}
    It is \NP-hard to determine the neighborhood number in co-bipartite graphs.
\end{theorem}

\begin{proof}
To prove that the problem is \NP-hard, we shall use Theorem~\ref{teo:inaproxOfVC}, and show that a polynomial-time reduction from a $\frac{4}{3}$-approximation of the Minimum Vertex Cover problem can be easily obtained. For that purpose, given a graph $H$, we will show to build a co-bipartite graph $G$ such that $\beta(H) \leq \pn{G} \leq \beta(H) + 1$. Namely, given any graph $H=(V,E)$, let $G=(X,Y,F)$ where $X = V$, $Y = E$ and $F = \{ve\in V\times E\colon v\text{ is incident to }e\text{ in }H\}$; that is, every vertex in $X$ is joined to the edges in $Y$ to which it is incident in $H$.

Given a set vertex cover $C\subseteq V$ of $H$, then $C$ together with any element of $Y$ is clearly neighborhood set of $G$. In fact, all the edges of the cliques $X$ and $Y$ will clearly be covered by any vertex of $X$ and the vertex of $Y$, respectively. Moreover all edges of $F$ will be covered because if $ve\in F$, then $e = vw$ (in $H$) for some $w\in V$. Hence, since $C$  was a vertex cover of $H$, $v$ or $w$ must be in $C$ and both cover the edge $ve$ in $G$ (because $v,w,e$ is a triangle in $G$). Thus $\pn{G} \leq \beta(H) + 1$. 

To check the remaining inequality, let $S \subseteq X\cup Y$ be a neighborhood set of $G$ with minimum cardinality. If $e$ is any element in $S \cap Y$,  then $e$ is covering in $G$ only two edges of $F$, namely the $ve$ and $we$, where $vw = e$ (in $H$). Thus if we replace $e$ by $v$ or $w$ in $S$, this set that arises still covers all the the edges of $F$. If we apply this procedure successively for all vertices in $S \cup Y$, we will obtain at the end a vertex set of $S'\subseteq X$ that is a neighborhood-covering set of $F$ and has size less than or equal to $\pn{G}$. It turns out that $S'\subseteq V$ will be a vertex cover of $H$, because for any edge $e\in E$, where $e=vw$ (in $H$), $v$ or $w$ will be in $S'$ for these are the only vertices in $X$ that cover $ve\in F$. As $S'$ is a a vertex cover of $H$ whose size is less than or equal to $\pn{G}$, $\beta(H) \leq \pn{G}$.

Now that we have proved that this co-bipartite graph $G$ satisfies $\beta(H) \leq \pn{G} \leq \beta(H) + 1$, it is easy to give a polynomial-time reduction to the problem of approximating $\beta(H)$ within a factor of $\frac{4}{3}$. Given a graph $H$, we can in polynomial (linear) time decide whether it has a vertex cover of size $1$ or $2$ and, if so, we transform $H$ into an arbitrary co-bipartite graph whose corresponding maximum neighborhood set has size $1$ or $2$, respectively. If $\beta(H) \geq 3$ we construct in polynomial time $G$ as described above. As proven before $\beta(H) \leq \pn{G} \leq \beta(H) + 1$, which as $\beta(H) \geq 3$ means that $ 1 \leq \frac{\pn{G}}{\beta(H)} \leq 1 + \frac{1}{3}$. This proves the reduction from the problem of approximating the Minimum Vertex Cover problem less than $10\sqrt{5}-21 = 1.3606\dots$, as desired.     
\end{proof} 

\section{Further Remarks} 
\label{sec:further_remarks}

It is worth noting that a different approach for obtaining linear-time algorithms for $P_4$-tidy graphs (and, more generally, in graph classes having bounded clique-width) was introduced by Courcelle et al.\ in \cite{CourcelleLinearTimeBoundedCW}. This approach allows for linear-time solutions of recognition and optimization problems that are expressible in a certain monadic second-order logic. Given the characterizations proven in Theorems~\ref{teo:CaractOfP4TidyNP} and \ref{teo:CaractNPTreeCograph}, it is easy to see that the recognition problem of neighborhood-perfectness in $P_4$-tidy graphs and tree-cographs can be expressed in this monadic second-order logic. As the class of tree-cographs also has bounded clique-width, Courcelle et al.'s metatheorem would imply the existence of a linear-time algorithm for the recognition problem of neighborhood-perfectness when the input graph is restricted to both tree-cographs and $P_4$-tidy graphs. Nevertheless, we stress the necessity of Theorems~\ref{teo:CaractOfP4TidyNP} and \ref{teo:CaractNPTreeCograph} for proving the existence of such linear-time recognition algorithms using the approach of \cite{CourcelleLinearTimeBoundedCW} since, while the fact of containing or not a particular induced subgraph is expressible in the corresponding monadic second-order logic, the property itself of being or not neighborhood-perfect seems not expressible at all in such logic. Moreover, the problem of finding a minimum neighborhood-covering set can be seen to fall as well in the scope of the approach of \cite{CourcelleLinearTimeBoundedCW}, implying a linear-time algorithm to solve this problem in any $P_4$-tidy graph or tree-cograph. Nevertheless, although Courcelle et al.'s metatheorem is of great theoretical importance, the algorithm obtained by it is far away from being practical; because it may have enormous hidden constants in the linear-time complexity (even if the input graph has small clique-width) \cite{Courcelle2008}. This combinatorial explosion of the constants seems to be a consequence of the generality of the metatheorem, given that it requires only a monadic second-order formula and an input graph to solve the problem. This seems unavoidable if one wishes to obtain results for general monadic second-order formulas \cite{FrickGrohe2004ComplexOfMSOL}. Therefore, it is clearly of interest to find more practical algorithms, that can work by only performing a simple transversal of the modular decomposition trees of the the input graph as those developed in Section~\ref{sec:algorithms_and_complexity_results}. In addition, linear-time algorithms for solving the neighborhood-independence number problem for $P_4$-tidy graphs and tree-cographs, as those given in Subsection~\ref{subsec:optimal}, do not seem to follow from the approach of \cite{CourcelleLinearTimeBoundedCW} as the problem is not directly expressible in the corresponding logic (as quantification over subsets of edges is not allowed).

\section*{Acknowledgments}
This work was partially supported by UBACyT Grant 20020130100808BA, CONICET PIP 112-201201-00450CO and PIO 14420140100027CO, ANPCyT PICT 2012-1324 (Argentina), FONDECyT Grant 1140787 and Millennium Science Institute Complex Engineering Systems (Chile).

\end{document}